%% file: arxiv.tex
\documentclass[9.5pt,journal,compsoc]{IEEEtran}

\usepackage{graphicx}
\usepackage{amsmath}
\usepackage{amssymb}
\usepackage{amsthm}
\usepackage{url}
\usepackage{xspace}
\usepackage{balance}
\usepackage{algorithm}
\usepackage{algpseudocode}
\usepackage{multirow}
\usepackage{subfig}
\usepackage{setspace}
\usepackage{hyperref}
\usepackage{enumitem}
\usepackage{fixltx2e}
\usepackage{xcolor}
\usepackage{times}
\usepackage{booktabs}

\graphicspath{{graphs/}}

\allowdisplaybreaks

\newtheorem{thm}{\textbf{Theorem}}

\newtheorem{prob}{\textbf{Problem}}

\setlist{nolistsep,leftmargin=*}

\title{TIPS: Mining Top-K Locations \\ to Minimize User-Inconvenience \\ for Trajectory-Aware Services}

\author{
\IEEEauthorblockN{
Shubhadip Mitra, 
Priya Saraf, 
Arnab Bhattacharya 
}\\
\IEEEauthorblockA{ Dept. of Computer Science and Engineering, Indian Institute of Technology, Kanpur, India. \\
\{smitr,priyas,arnabb\}@cse.iitk.ac.in}
}

\newcommand{\ds}{\ensuremath{d}\xspace}
\newcommand{\dr}{\ensuremath{d_r}\xspace}

\newcommand{\tops}{TOPS\xspace}

\newcommand{\tips}{TIPS\xspace}
\newcommand{\maxtips}{MAX-TIPS\xspace}
\newcommand{\meantips}{AVG-TIPS\xspace}

\newcommand{\traj}{\ensuremath{\mathcal{T}}\xspace}
\newcommand{\fac}{\ensuremath{\mathcal{F}}\xspace}
\newcommand{\subtraj}{\ensuremath{\mathcal{T}^\prime}\xspace}
\newcommand{\reptraj}{\ensuremath{\mathcal{RT}}\xspace}

\newcommand{\s}{\ensuremath{\mathcal{S}}\xspace}
\newcommand{\subs}{\ensuremath{\mathcal{S}^\prime}\xspace}
\newcommand{\q}{\ensuremath{\mathcal{Q}}\xspace}

\newcommand{\uf}{user-fraction\xspace}
\newcommand{\ufs}{\ensuremath{\gamma}\xspace}

\newcommand{\maxi}{\ensuremath{M \negthinspace I}\xspace}
\newcommand{\totali}{\ensuremath{T \negthinspace I}\xspace}
\newcommand{\ai}{\ensuremath{A \negthinspace I}\xspace}

\newcommand{\mif}{MIF\xspace}
\newcommand{\nc}{NetClus\xspace}

\newcommand{\hcc}{HCC\xspace}
\newcommand{\incg}{GREAT\xspace}
\newcommand{\inc}{GREAT\xspace}

\newcommand{\opt}{Opt\xspace}

\newcommand{\clus}{Clus\xspace}

\newcommand{\nf}{swap-fraction\xspace}
\newcommand{\nfs}{\ensuremath{s \negthinspace f}\xspace}
\newcommand{\nn}{\ensuremath{\mathcal{NF}}\xspace}
\newcommand{\trajthres}{\ensuremath{\ufs  \times |\traj|}\xspace}

\newcommand{\bs}{Beijing-Small\xspace}
\newcommand{\bm}{Beijing-Medium\xspace}
\newcommand{\bl}{Beijing-Large\xspace}
\newcommand{\bms}{Beijing-Medium-Sampled\xspace}
\newcommand{\bls}{Beijing-Large-Sampled\xspace}

\newcommand{\subfigwidth}{0.49\columnwidth}

\newcommand{\figcaption}[1]{\vspace*{-2mm}\caption{#1}\vspace*{-0.5mm}}
\newcommand{\tabcaption}[1]{\vspace*{-2mm}\caption{#1}\vspace*{-0.5mm}}

\newcommand{\moveup}{\vspace*{-2mm}}
\newcommand{\moveups}{\vspace*{-3mm}}
\newcommand{\eps}{\ensuremath{\epsilon}\xspace}
\newcommand{\epsinc}{\ensuremath{(1 + \epsilon)}\xspace}

\begin{document}

\setlength{\abovedisplayskip}{3pt}
\setlength{\belowdisplayskip}{3pt}
\setlength{\abovedisplayshortskip}{3pt}
\setlength{\belowdisplayshortskip}{3pt}

\maketitle

\begin{abstract} 
	\emph{Facility location} problems aim to identify the best locations to set
	up new services. Majority of the existing works typically assume that the
	users are \emph{static}. However, there exists a wide array of services
	such as fuel stations, ATMs, food joints, etc., that are widely accessed by
	mobile users besides the static ones.  Such \emph{trajectory-aware
	services} should, therefore, factor in the trajectories of its users rather
	than simply their static locations.  In this work, we introduce the problem
	of optimal placement of facility locations for such trajectory-aware
	services that \emph{minimize the user inconvenience}.  The inconvenience of
	a user is the \emph{extra} distance traveled by her from her regular path
	to avail a service.  We call this the \tips problem \emph{(Trajectory-aware
	Inconvenience-minimizing Placement of Services)} and consider two variants
	of it.  The goal of the first variant, \maxtips, is to minimize the
	\emph{maximum} inconvenience faced by any user, while that of the second,
	\meantips, is to minimize the \emph{average} inconvenience over all the
	users. We show that both these problems are NP-hard, and propose multiple
	efficient heuristics to solve them.  Empirical evaluation on real
	urban-scale road networks validate the efficiency and effectiveness of the
	proposed heuristics.
\end{abstract}

\begin{IEEEkeywords}
Trajectory-aware service, User inconvenience, \tips.
\end{IEEEkeywords}

\input{intro}

\input{related}

\input{formulation}

\input{maxtips}

\input{meantips}

\input{results}

\input{conc}

\pagebreak

\bibliographystyle{IEEEtran}
\balance
\bibliography{../papers}

\pagebreak

\appendices

\input{appendix}

\balance

\end{document}

%% file: intro.tex
\section{Introduction}
\label{sec:intro} 

\emph{Facility location} problems identify the best locations to set up new
facilities (or services) for its users \cite{drezner1995facility,
FacilityLocation}.  This has been also studied as \emph{optimal location}
problems \cite{du2005optimal, zhang2006progressive, ghaemi2010optimal,
xiao2011optimal, chen2014efficient}.  Majority of the existing works, however,
assume that the users of the service are \emph{static}
\cite{qi2012min,qi2014min,liu2016finding}.  However, services such as fuel
stations, automobile service stations, ATMs, food joints, convenience stores,
etc., are widely accessed by mobile users, besides the static users
\cite{tops_paper}. For example, it is common for many users to make their daily
purchases while returning from their  workplaces. This practice is increasingly
becoming common because of rapid expansion of cities due to growing urban
population and consequent longer work commute trips. The placement of such
services should therefore take into account the \emph{mobility patterns} or the
\emph{trajectories} of its users, rather than simply their static home and
office locations.  The necessity for such \emph{mobility-aware location}
selection has also been highlighted in recent studies
\cite{tops_paper,MILERUN,PINOCCHIO}. Moreover, our experimental findings in
Sec.~\ref{sec:baseline comparison}  suggest that there is 10-40 \% cost-savings
if we factor in the user-trajectories instead of their static locations.

A \emph{user trajectory} is a sequence of spatial points that lie on the path
of a user while travelling. It is important to note that trajectories strictly
generalize the static user scenario as static users can always be modeled as
trajectories with a single location.  In general, however, trajectories capture
user location patterns more effectively and realistically.

In this work, we extend two key optimal location problems, namely the
\emph{MinMax Location Query}
\cite{xiao2011optimal,chen2014efficient,liu2016finding,TODS:OLQRN} and the
\emph{Min-Dist Location Query}
\cite{qi2012min,qi2014min,MILERUN,Qi:2012:BBM:2483739.2483747} (also referred
to as the \emph{MinSum Location Query} \cite{VLDB:OLQRN}). Given a set of
customers (or users) $\mathcal{C}$, a set of existing facilities $\mathcal{F}$,
and a set of candidate locations that can host a new facility $\mathcal{S}$,
the goal of the MinMax Location Query (respectively, Min-Dist Location Query)
is to identify a facility location in $\mathcal{S}$ that minimizes the maximum
(respectively, average) distance of any user to its nearest facility.

Majority of the existing works assume that the users are \emph{static} and
ignore their mobile behavior. Further, most of these works also restrict
themselves to reporting a \emph{single} facility location, which is
polynomially solvable. Motivated by these two limitations, in this work, we
introduce two novel optimal location problems, \emph{\maxtips} and
\emph{\meantips}, that factor the user trajectories and report any desired
number of facility locations.

Formally, given a set of trajectories \traj, a set of existing facilities
$\mathcal{F}$, a set of candidate sites \s, an integer $k$, and a \uf $\ufs \in
[0,1]$, the \maxtips problem seeks to report a set $\q \subseteq \s$ of $k$
locations that minimizes the \emph{maximum inconvenience} over any $\ufs$
fraction of the trajectories \traj, while the \meantips problem aims to identify
the $k$ locations that minimize the \emph{average inconvenience} faced by any
user.  The \emph{inconvenience} of a user on a trajectory is defined as the
\emph{extra distance} traveled by her with respect to her normal trajectory to
avail the service.  The proposed problems are NP-hard.

While for critical services such as ambulance or fire stations, it is desirable
to minimize the \emph{maximum inconvenience}, for other services such as ATMs,
fuel stations or convenience stores, it is desirable to minimize the
\emph{average inconvenience}.  Both these \tips problems have applications in
various resource planning scenarios
\cite{Qi:2012:BBM:2483739.2483747,VLDB:OLQRN,Khan:2018:GOAL}.
Some of the direct applications are:

\noindent \emph{Placement of drop-boxes for crowd-sourced taxi shipment
service:} Chen et al. \cite{TaxiExp} presented a novel scheme for city-wide
shipment of items using the regular passenger-carrying taxis in a crowd-sourced
manner. The taxis participating in this service are required to collect and
drop the shipments at a nearby drop-box whenever they are idle, i.e., not
carrying any passenger. The drop-boxes must be located in a
manner such that the maximum inconvenience faced by majority of the taxis is
minimal.

\noindent \emph{Locating multiple ATMs of a given bank:} Suppose a bank plans
to set up $k$ ATMs in a given city.  Since mobile users often access nearby
ATMs, they must be placed such that the average user-inconvenience is
minimized.

\noindent \emph{Placement of food trucks of a given chain:} Extending the
example given in \cite{Khan:2018:GOAL}, consider a restaurant chain that wants
to place its $k$ food trucks on the road network, so as to serve the mobile
customers (who have shared their trips) such that the average inconvenience
faced by any user is minimal.

We illustrate the \tips problems through an example shown in
Fig.~\ref{fig:example}. There are $6$ trajectories $T_1,\dots,T_6$ (shown with
blue dashed lines with arrows indicating the directions of the respective trip),
one existing facility at $s_0$ (marked in green) and $4$ candidate sites
$s_1,\dots,s_4$ (shown in red) to host a new facility.  The road segments are
marked in black with corresponding distances (assumed to be the same on both
ways). The nodes $v_1$ and $v_2$ (shown in black) are general points on the road
network that do not host a facility.  If the user on trajectory $T_1$ (that
passes through $v_1$ and $s_1$) wishes to access the facility at $s_0$, she
needs to detour from $v_1$, visit $s_0$, and join her regular path at $s_1$.  As
a result, her inconvenience (i.e., the extra distance traveled) is $1+2-2=1$
unit. Now if another facility comes up at $s_1$, her inconvenience reduces to
$0$, as there is no detour. For trajectories $T_2,T_3$, the inconvenience w.r.t.
$s_0$ is $7+2+2+7=18$ units (as they need to take a round trip via
$s_2,s_1,s_0,s_1,s_2$, in that order). Similarly, for trajectories $T_4,T_5$ and
$T_6$, it is $30,32$ and $20$ respectively.  Note that the inconvenience is
measured w.r.t. the nearest facility from any point on the trajectory.

Assume that the service-provider wants to set up $2$ new facilities besides the
existing one at $s_0$, with the objective to minimize the \emph{maximum}
inconvenience faced by any user. If the new facilities are hosted at
$\{s_1,s_2\}$, the inconvenience of trajectories $T_1,T_2,T_3$ are $0$ units
each.  For trajectories $T_4$ and $T_5$ the nearest facility is $s_2$ and,
therefore, their inconveniences are $12$ units each.  Similarly, for trajectory
$T_6$, the nearest facility is $s_1$ and, thus, its inconvenience is $16$
units.  Therefore, the maximum inconvenience among all the trajectories due to
the selection $\{s_0,s_1,s_2\}$ is $16$ units.  The maximum inconvenience for
all such selections of $2$ new sites (along with $s_0$) are listed in
Fig.~\ref{fig:example} under the column $\ufs=1$. (We will shortly explain the
meaning of $\ufs$.) The selection $\{s_0,s_3,s_4\}$ offers the optimal maximum
inconvenience of $12$ units.  Importantly, although most number of trajectories
pass through $s_2$, it is \emph{not} part of the optimal solution.

Next, suppose the objective is to minimize the \emph{average} (or equivalently,
the \emph{total}) inconvenience over all the user trajectories.  The last
column in the table in Fig.~\ref{fig:example} lists the total inconvenience for
all possible selections.  The selection $\{s_0,s_2,s_3\}$ offers the optimal
total inconvenience of $21$ units. Thus, the optimal average inconvenience is
$21/6=3.5$ units.  (The optimal maximum inconvenience was $12$ units.)

\begin{figure}[t]
  \centering
  \includegraphics[width=0.30\columnwidth]{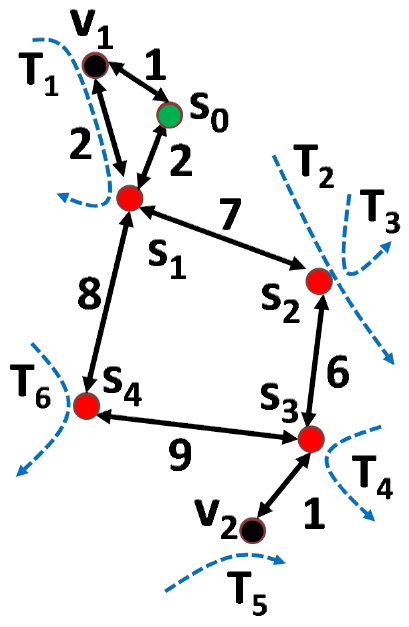}
  \includegraphics[width=0.60\columnwidth]{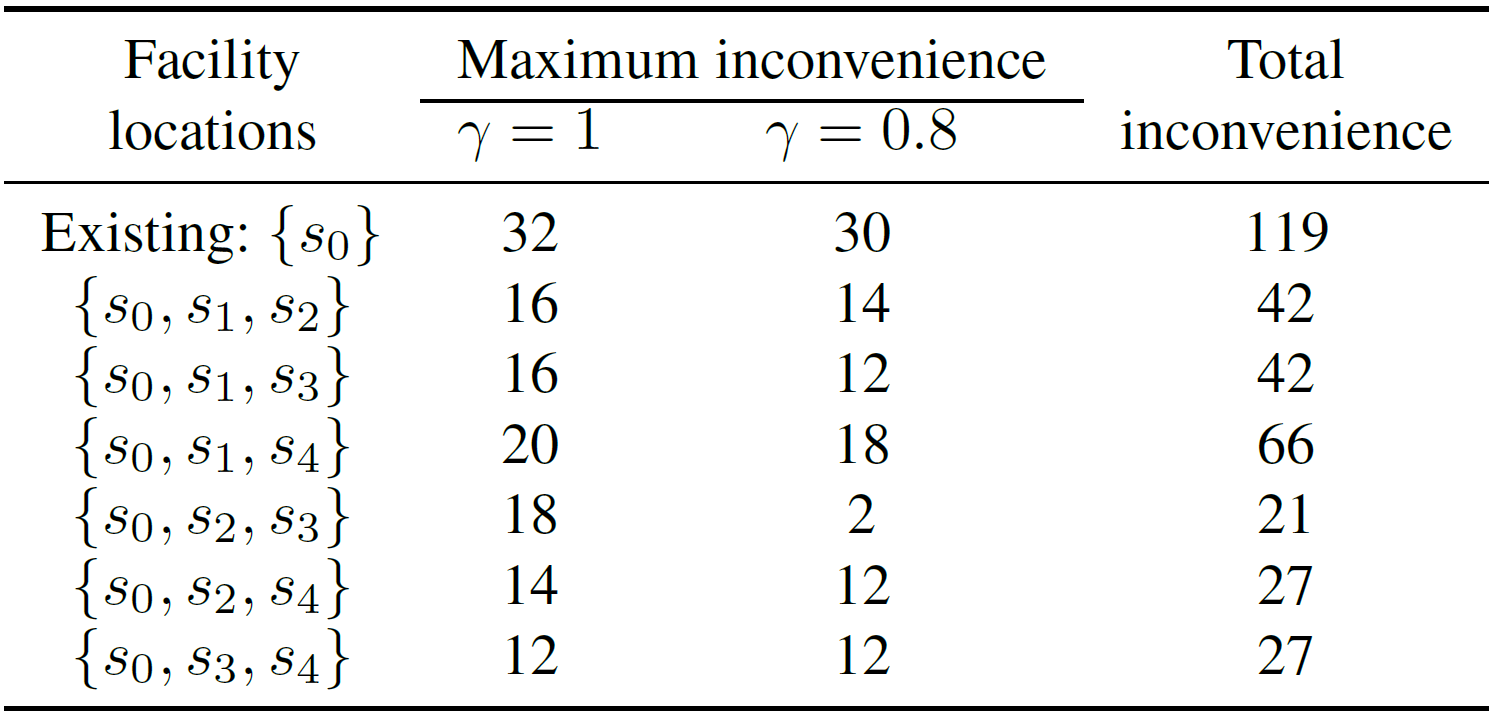}
  \figcaption{Illustration of the need for minimizing user-inconvenience for
  trajectory-aware services.}
  \label{fig:example}
\end{figure}

The maximum inconvenience problem suffers from the issue of outlier
trajectories where a trajectory is very different from all the other ones and,
therefore, accommodating for it becomes harder.  Thus, instead of considering
all the trajectories, the service provider may choose a fraction of the
trajectories, over which the maximum inconvenience will be minimized.  We call
this fraction the \emph{\uf}, \ufs.

Referring to Fig.~\ref{fig:example}, when $\ufs=0.8$, the goal is to minimize
the maximum inconvenience over at least $0.8 \times 6 = 4.8$ or $5$
trajectories. The values of the maximum inconvenience for all the selections
for $\ufs=0.8$ are listed in the table. Note that the optimal selection for
$\ufs=0.8$ is $\{s_0,s_2,s_3\}$ which is different from the optimal selection
for $\ufs=1$. Note that the optimal inconvenience for $\gamma = 0.8$ falls to
$2$ units as compared to $12$ units for $\ufs=1$.

A naive approach to solve either of the \tips problems involves enumerating all
$k$-sized subsets of $\mathcal{S}$, computing the maximum or the average
inconvenience (depending on the objective) for each subset, and returning the
subset with the minimal objective value. This requires an exponential time and
space complexity, which is infeasible for almost all datasets.  Thus, the major
challenges in solving \tips are as follows:

\begin{enumerate}

	\item Quality: Since both the problems are NP-hard (proved later), optimal
		algorithms are impractical.  Thus, we need efficient heuristics that
		offer high quality solutions.

	\item Scalability: Any basic approach to solve the above problems would
		typically need to compute and store pairwise distances between the sets
		of candidate sites and trajectories.  However, for any city-scale
		datasets, this time and storage requirement is prohibitively large (of
		the order of 100s of GB).  Thus, it is necessary to design solutions
		that are practical and scalable.
	
\end{enumerate}

In this paper, we propose efficient heuristics for both \maxtips and \meantips
that overcome the above challenges.  To summarize, our major contributions are:

\begin{enumerate}

	\item To the best of our knowledge, this is the \emph{first} work, that
		factors in user-mobility to identify the best $k$ facility locations
		that minimize the maximum and the average user-inconvenience, in
		presence or absence of \emph{existing facilities}. In particular, we
		introduce two \emph{facility location problems} over user trajectories,
		namely, \maxtips and \meantips (Sec.~\ref{sec:formulation}). 
		
	\item We show that both these problems are \emph{NP-hard}
		(Th.~\ref{thm:NP-hardness}), and propose one exact algorithm and two
		polynomial-time efficient \emph{heuristics} to solve each of them
		(Sec.~\ref{sec:maxtips} and Sec.~\ref{sec:meantips}).  The exact
		algorithms are based on integer linear programming
		(Sec.\ref{sec:opti-max} and \ref{sec:opti-mean}).  For \maxtips, while
		the first heuristic is an index-free greedy algorithm
		(Sec.~\ref{sec:mif}), the second one (Sec.~\ref{sec:nc}) uses an index
		structure based on multi-resolution clustering of the road network.
		For \meantips, the first heuristic is a hybrid of two standard
		clustering algorithms that are based on local search techniques
		(Sec.~\ref{sec:hcc}), while the second one (Sec.~\ref{sec:great}) uses
		a simple greedy approach.

	\item Empirical evaluation on urban scale datasets show that our heuristics
		are \emph{effective} in terms of quality, and \emph{efficient} in terms
		of space and running time (Sec.~\ref{sec:exp}).

\end{enumerate}

%% file: related.tex
\section{Related Work}
\label{sec:related}

\setlength{\tabcolsep}{3pt}	

\begin{table*}[t]
	\scriptsize
	\centering
	\begin{tabular}{|c||c|c|c|c|c|c|c|c|c|c|c|c|c|c|c|c|c|c|c|}
		\hline
		{\bf Property} & \cite{du2005optimal} & \cite{zhang2006progressive} & \cite{ghaemi2010optimal} & \cite{xiao2011optimal} & \cite{chen2014efficient} & \cite{qi2012min} & \cite{qi2014min} & \cite{liu2016finding} & \cite{tops_paper} & \cite{MILERUN} & \cite{PINOCCHIO} & \cite{TODS:OLQRN} & \cite{Qi:2012:BBM:2483739.2483747} & \cite{VLDB:OLQRN} & \cite{Khan:2018:GOAL} & \cite{Wong:2009:VLDB:MBRNN} & \cite{OMPQuery2015} & \cite{CollectiveKOptimal} & \cite{li2013trajectory} \\
		\hline
		\hline
		Users              & Stat & Stat & Stat & Stat & Stat & Stat & Stat & Stat & Mob & Mob & Mob & Stat & Stat & Stat & Mob & Stat & Stat & Stat & Mob \\
		\hline
		Output             & Sngl & Sngl & Sngl & Sngl & Sngl+Top-k & Sngl & Sngl & Top-k & Top-k & Sngl & Sngl & Sngl+Top-k & Sngl & Sngl & Sngl & Sngl & Sngl & Top-k & Sngl \\
		\hline
		Objective function & Infl & Dist & Infl & Dist & Dist & Dist & Dist & Dist & Infl & Dist & Infl & Dist & Dist & Infl & Dist & Infl & Dist & Dist & Infl \\
		\hline
		Underlying space   & Eucl & Eucl & Road & Road & Road & Eucl & Eucl & Road & Road & Road & Eucl & Road & Eucl & Road & Eucl & Eucl & Road+Eucl & Eucl & Road \\
		\hline
	\end{tabular}
	\tabcaption{Summary of Related Work. (Stat: static, Mob: mobile; Sngl: single, Top-k: top-k; Dist: distance, Infl: influence; Eucl: Euclidean space, Road: road network.)}
	\label{tab:related}
\end{table*}

\setlength{\tabcolsep}{6pt}	

The related work is divided into the following main classes.

\noindent
\textbf{Location Selection (LS) Queries:}
Location Selection (LS) queries identify best locations to set up new
facilities. They are broadly of two types: Optimal Location (OL) queries
\cite{du2005optimal} and Facility Location (FL) Queries \cite{FacilityLocation}.
An OL query typically has three inputs, a set of existing facilities, a set of
users, and a set of candidate sites, and the aim is to find a candidate site to
host a new facility that optimizes an objective function based on the distances
among the facilities and the users.  On the other hand, an FL query considers a
set of users, and a set of candidate sites, and seeks to identify $k$ ($k \ge
1$) facility locations that optimize certain objective function which is usually
based on the distances among the users and the facilities. While OL queries are
polynomially solvable, FL queries are NP-hard. In Table~\ref{tab:related}, we
summarize the key related works in the area of location selection based on the
following attributes.

\noindent \textbf{(1)~Type of Users:} Majority of the early works assumed that
the users are fixed to a single location such as home, and did not consider
their mobile behavior. However, many recent studies including ours factor in
user-mobility.

\noindent \textbf{(2)~Type of Output:}  While many works report a single
location to set up a new facility, others report top-$k$ facility locations
that collectively optimize the desired objective. While the former is
polynomially solvable, the latter is NP-hard. Both the \tips problems return
$k$ locations, and are NP-hard as well.

\noindent \textbf{(3)~Type of Objective:} Based on the objective function, the
related works can be classified into two categories: \emph{distance minimizing}
and \emph{influence maximizing}.  In distance minimizing queries, the goal is to
minimize the maximum or the average distance of any user to its nearest
facility. While the former is referred to as the MinMax Location Query
\cite{xiao2011optimal, chen2014efficient, liu2016finding, TODS:OLQRN},  the
latter is known as the Min-Dist Location Query or MinSum Location Query
\cite{qi2012min, qi2014min, MILERUN, Qi:2012:BBM:2483739.2483747, VLDB:OLQRN}.
\maxtips problem is a generalization of the MinMax Location query, while
\meantips is a generalization of the Min-Dist Location query. These
generalizations are in terms of reporting $k$ locations instead of a single
location, and considering mobile users on a road network instead of static ones.
The aim of influence maximizing queries is to find a candidate site that has
maximal \emph{influence} over its users. The influence of a site $s$ is the
number of users for which $s$ is the nearest facility.  These problems are
usually modeled as \emph{Reverse Nearest Neighbor} queries. The key difference
between these works and ours is that they assume that the new facility is
\emph{competing} with the existing ones; in our model, both the new and the
existing facilities (if any) belong to a given service provider and, hence,
complement each other. This is especially true for public services such as ATMs,
post offices, hospitals, gas stations, parking spots, etc.

\noindent \textbf{(4)~Type of Underlying Space:} Many earlier models assume that
the underlying space is Euclidean. Since user movements are typically restricted
by a road network, and network distances can significantly vary from
corresponding Euclidean distances, recent works base their studies on road
networks. Our \tips formulation is also based on the road network. The massive
distance computations, and absence of geometric properties make the latter
problems more challenging.

Referring to Table~\ref{tab:related}, we note that this is the \emph{first}
work that factors in \emph{user-mobility} to mine the \emph{top-k} facility
locations that minimizes the \emph{inconvenience} (defined in terms of the
distances between the users and facilities) caused to the users traveling on a
\emph{road network}.  Next, we discuss the important LS works that are closely
related to our proposed \tips query.

\noindent $\bullet$ \textbf{Key related works factoring in user-mobility:}
In \cite{PINOCCHIO}, the authors studied the PRIME-LS problem that aims to find
an optimal location which can influence the most number of moving objects. Two
algorithms were proposed based on two pruning techniques and optimization
strategies to filter out unpromising candidate sites. However, since the goal is
to maximize the influence rather than to minimize the distance, the attributes
of this problem are quite different to ours (as explained earlier). Thus, these
algorithms are inapplicable to solve the \tips problems. The studies in
\cite{tops_paper,li2013trajectory} consider the FL problem over user
trajectories moving on a road network. They assume that a user is attracted to a
facility if its trajectory lies within a specified distance threshold. While
both these works aim to maximize the user coverage, \cite{li2013trajectory}
reports a \emph{single} optimal road segment, and \cite{tops_paper} reports the
$k$ best facility locations.

MinMax Location Query and Min-Dist Location Query have been studied in
\cite{zhang2006progressive, xiao2011optimal, chen2014efficient, qi2012min,
qi2014min, liu2016finding, MILERUN, TODS:OLQRN, Qi:2012:BBM:2483739.2483747,
VLDB:OLQRN, Khan:2018:GOAL, OMPQuery2015}.  All other works except
\cite{MILERUN,Khan:2018:GOAL} assume the users to be static.  However, in
contrast to our work, both these works find a \emph{single} optimal location for
the min-dist location problem over user trajectories. Moreover, since
\cite{Khan:2018:GOAL} study the problem over Euclidean space, their techniques
are not applicable to our model which is based on a road network. The study of
\cite{MILERUN} is however based on a road network. Based on reference location
transformation, they propose two groups of algorithms. While the first group
uses spatial locality based index structures, the other group does not use any
index structure but computes from scratch.  In our empirical study, we consider
this algorithm as a baseline. 

Majority of the works in the FL literature also assume that the users are static
\cite{drezner1995facility, FacilityLocation}. Works that consider human mobility
include \cite{hodgson1981location, berman1992optimal,
berman1995locatingDiscretionary, berman1995locatingMain, berman1995locating,
berman2002generalized, berman1998flow}.  These works, however, assume a flow
model to characterize mobility instead of using real trajectories.  The proposed
models are mostly theoretical and are not scalable for real city-scale road
networks. In particular, all these approaches require extensive distance
computations which leads to large memory overhead and are, hence, infeasible
\cite{tops_paper}.  A fairly comprehensive literature survey is available in
\cite{FIFLPSurvey}.

Hodgson \cite{hodgson1981location} posed the first FL problem that minimizes
the average user inconvenience as a generalization of $k$-medians problem. In
\cite{berman1995locatingMain}, it was shown that the problem does not admit
constant factor approximation unless $P=NP$. However, no approximation
algorithm was proposed. In contrast, we propose two heuristics for this
problem.

\noindent
\textbf{Clustering Problems:} The following clustering problems are related to
the current work.

\noindent
\textbf{(1) $k$-Center Problem:} Given a set $\s$ of $n$ points, the
\emph{k-center} problem is to determine a set $\q \subseteq \s$ of size $k$,
referred to as \emph{centers}, such that the maximum distance of any point in \s
to its nearest center is minimized. This problem was introduced and proved to be
NP-hard in \cite{gonzalez1985clustering}. They also proposed a greedy heuristic
that offers a factor of $2$ approximation.
Our proposed \maxtips problem is a generalization of the $k$-center problem as
trajectories generalize static users.  We have extended the greedy algorithm in
\cite{gonzalez1985clustering} to design a heuristic (with bounded quality
guarantees) to solve \maxtips (Sec.~\ref{sec:maxtips}).

\noindent
\textbf{(2) $k$-Medoids Problem:} Given a set $\s$ of $n$ points, the
\emph{k-medoids} problem is to determine a set $\q \subseteq \s$ of size $k$,
referred to as \emph{medoids}, such that the sum of distances of each point in
\s to its nearest medoid is minimized \cite{rousseeuw1990finding}.  This is also
referred to as the \emph{$k$-median problem}.  The first constant factor
algorithm for the $k$-median problem in general metric space, with an
approximation ratio of $6\frac{2}{3}$ was proposed in \cite{guha1999kmedian}.
Later, \cite{jain1999} improved this factor to $6$.  Subsequently,
\cite{moses1999} designed a $4$-approximation algorithm that runs in $O(n^3)$
time.  Korupolu et al.  \cite{korupolu2000analysis} proposed a local search
based approximation scheme by allowing a constant factor blow-up in $k$.  Arya
et al.  \cite{arya2004local} improved the approximation bound for the metric
$k$-medians problem to $3+2/p$ where $p$ is the number of medians swapped
simultaneously. Three popular local search based techniques for the $k$-medoids
problem are PAM \cite{rousseeuw1990finding}, CLARA \cite{rousseeuw1990finding}
and CLARANS \cite{ng2002clarans}. These schemes are detailed in
Sec.~\ref{sec:hcc}.  Most of these works, however, fail to scale on large
datasets.  Moreover, all these algorithms require the distance matrix between
each pair of points. This quadratic space and time requirement renders them
infeasible for real datasets.  Our proposed \meantips problem generalizes the
$k$-medoids problem as explained in Sec.~\ref{sec:formulation}.  The \hcc
heuristic to solve \meantips (Sec.~\ref{sec:hcc}) builds on the ideas of CLARA
and CLARANS. More importantly, to address the challenge of high space overhead,
we design sampling techniques to reduce the sizes of the sets.

%% file: formulation.tex
\section{The \tips Problem}
\label{sec:formulation}

Consider a road network $G = \{V,E\}$ over a geographical area where
$V=\{v_1,\dots,v_N\}$ denotes the set of road intersections, and $E$ denotes
the road segments between two adjacent road intersections.  To model the
direction of the underlying traffic that passes over a road segment, we assume
that the edges are directed.  Assume a set of candidate sites $\s = \{s_1,
\cdots, s_n\}$ where a certain service or facility can be set up.  The set $\s$
can be in addition to the existing facility locations \fac.  Without loss of
generality, we can augment the vertices $V$ to include all the sites. Thus, $\s
\subseteq V$. Further, we also assume that the set of existing facilities $\fac
\subseteq V$.

The set of trajectories is denoted by $\traj = \{T_1, \cdots, T_m\}$
where each trajectory $T_j = \{v_{j_1}, \cdots, v_{j_l}\}$, $v_{j_i}\in V$, is
a sequence of locations that the user passes through.  Usually any service or
facility is used by a mix of static users and mobile users. As the above
definition of trajectory allows both single and multiple locations of a single
user, it captures both static and mobile users simultaneously.  The
trajectories are usually recorded as GPS traces and may contain arbitrary
spatial points on the road network.  For our purpose, each trajectory is
map-matched \cite{map1} to form a sequence of road intersections through which
it passes.

We also assume that each trajectory belongs to a separate user. The
framework can be easily generalized to multiple trajectories belonging to a single
user.  The union of road intersections that the user passes through in any of
her trajectories will be treated as the nodes in her trajectory.  In effect,
this will minimize her inconvenience from any of her trajectories.

Suppose $\ds(v_i,v_j)$ denotes the shortest road network distance along a
directed path from node $v_i$ to $v_j$, and $\dr(v_i,v_j)$ denotes the shortest
distance of a round-trip starting at node $v_i$, visiting $v_j$, and returning
to $v_i$, i.e., $\dr(v_i,v_j) = \ds(v_i,v_j) + \ds(v_j,v_i)$.  In general,
$\ds(v_i,v_j) \neq \ds(v_j,v_i)$, but $\dr(v_i,v_j) = \dr(v_j,v_i)$.  

The extra distance traveled by any user on trajectory $T_j$ to avail a service
at site $s_i \in V$, denoted by $\dr(T_j,s_i)$, is defined as follows:
$\dr(T_j,s_i) = \min_{\forall v_k,v_l \in T_j} \{ \ds(v_k,s_i) + \ds(s_i,v_l) -
\ds(v_k,v_l)\}$, i.e., it deviates from its trajectory at $v_k$, reach site
$s_i$ and then return to $v_l \in T_j$ such that the deviation is minimum. 

The round-trip distance between two trajectories $T_i$ and $T_j$ is defined as
the minimum pairwise distance among its sites: $\dr(T_i,T_j) =$ $\min_{\forall
v_i \in T_i, \ \forall v_j \in T_j} \{ \dr(v_i,v_j) \}$.  Henceforth,
\emph{distance} implies \emph{round-trip distance} unless mentioned otherwise.
 
It is inconvenient for a user to avail a service if the nearest service
location is far off from her trajectory.  We define this inconvenience as
follows.  Given a set of service locations $\q \subseteq \s$, the
\emph{inconvenience} of a user on trajectory $T_j$, denoted by $I_j$, is the
extra distance travelled to avail a service at the nearest service location in
$\q$. Formally, $I_j = \min\{ \dr(T_j,s_i)|s_i \in \q \}$.  

Using the above setting, we introduce two novel problems, namely \maxtips and
\meantips, described as follows. 

The \maxtips problem aims to report a set of $k$ service locations that
minimizes the maximum inconvenience over a given \emph{\uf} of the set of
trajectories.

\vspace*{-2mm}
\begin{prob}[\maxtips]
	Given a set of trajectories $\traj$, a set of existing facilities \fac, a set of candidate sites $\s$ that can
	host the services, a positive integer $k$, and a \uf \ufs ($0 < \ufs \leq
	1$), the \maxtips problem seeks to report a set $\q \subseteq \s, \ |\q| =
	k$, that minimizes the \emph{maximum inconvenience} over any set $\subtraj
	\subseteq \traj$ such that $|\subtraj| \geq \ufs \times |\traj|$, i.e., it
	minimizes $\maxi(\q) = \max_{T_j \in \subtraj} \{ I_j \}$, where $I_j =
	\min_{s_i \in \q \cup \fac} \{ \dr(T_j,s_i) \}$.
\end{prob}
\vspace*{-2mm}

Intuitively, as the \uf increases, the optimal value of $\maxi(\q)$ increases
because of the need to serve more number of users with the same number of $k$
facilities. When $\ufs=1$, the goal is to serve \emph{all} the trajectories
such that the maximum inconvenience faced by any trajectory is minimized.

The \meantips problem seeks to identify $k$ service locations that minimizes
the expected or average inconvenience across all the trajectories.  Since the
number of trajectories is fixed, minimizing the average inconvenience is
equivalent to minimizing the total inconvenience over all the trajectories.

\vspace*{-2mm}
\begin{prob}[\meantips]
	Given a set of trajectories $\traj$, a set of existing facilities \fac, a set of candidate sites $\s$ that can
	host the services, and a positive integer $k$, the \meantips problem seeks
	to report a set $\q \subseteq \s, \ |\q| = k$, that minimizes the
	\emph{total inconvenience} over \traj, i.e., it minimizes $\totali(\q) =
	\sum_{T_j \in \traj} I_j$, where $I_j = \min_{s_i \in \q \cup \fac} \{ \dr(T_j,s_i)
	\}$.
\end{prob}
\vspace*{-2mm}

We next show that both the \tips problems are NP-hard. 

\begin{thm}[NP-hardness of \tips]
\label{thm:NP-hardness}	
	\maxtips and \meantips are NP-hard problems.
\end{thm}

\vspace*{-2mm}
\begin{proof}
	Since the $k$-center problem is NP-hard \cite{gonzalez1985clustering} and
	it reduces to the \maxtips problem with each trajectory being a single user
	location, the set of existing facilities $\fac=\varnothing$ and $\ufs=1$,
	\maxtips is also NP-hard.

	Since the $k$-medoids problem is NP-hard \cite{rousseeuw1990finding} and it
	reduces to the \meantips problem with each trajectory being a single
	user-location and $\fac=\varnothing$, \meantips is also NP-hard.
	\hfill{}
\end{proof}

%% file: maxtips.tex
\section{Algorithms for \maxtips}
\label{sec:maxtips}

In this section, we present an optimal algorithm and two heuristics to solve
the \maxtips problem.

\subsection{Optimal Algorithm}
\label{sec:opti-max}

We present an \emph{optimal} solution to the \maxtips problem in the form of an
\emph{integer linear program} (ILP). For the ease of representation of the ILP,
we assume that the set of candidate sites $\s=\{s_i|1 \le i \le n\}$ is
augmented with the set of existing facilities $\fac$.
\begin{small}
\begin{align}
	\text{minimize } Z  &\quad \text{ such that } \label{eq:obj} \\
	\forall 1 \le i \le n, \ Z &\ge z_i   \label{eq:minmax constraint}\\		
	\forall 1 \le i \le n,\ \forall 1 \le j \le m, \ z_i &\ge   \dr(T_j,s_i) \times  y_{ij}  \label{eq:MaxDistanceConstraint}\\
	\sum_{i=1}^n x_i &\leq k +|\fac|, \label{eq:cardinality} \\
	\forall 1 \le j \le m,\ \sum_{i=1}^n y_{ij} &
	\le 1  \label{eq:trajclus1}\\
	\sum_{j=1}^m \sum_{i=1}^n y_{ij} &
	\ge \ufs \times |\traj| \label{eq:total_traj_served}\\
	\forall 1 \le i \le n,\ \forall 1 \le j \le m,\ y_{ij} &\le x_i  \label{eq:trajclus2} \\
	\forall 1 \le i \le n, \ 
	x_i &\in \{0,1\} \label{eq:binary1} \\
	\forall s_i \in \fac, \ x_i &=1 \label{eq:existing} \\
	\forall 1 \le i \le n,\ \forall 1 \le j \le m, \ 
	y_{ij} &\in \{0,1\}
	\label{eq:binary2}	
\end{align}
\end{small}

The Boolean variable $x_i=1$ if and only if the site $s_i$ is selected, or it
is an existing facility. The Boolean variable $y_{ij}=1$ if and only if the
site $s_i$ is a serving facility (either existing or new), and the trajectory
$T_j$ is served by $s_i$. The variable $z_i$ captures the maximum inconvenience
offered to any trajectory served by the site $s_i$
(Ineq.~\eqref{eq:MaxDistanceConstraint}). The constraint in
Ineq~\eqref{eq:minmax constraint}, along with the objective function in
Eq.~\eqref{eq:obj} ensure that the maximum inconvenience is minimized over the
set of selected sites. The constraint in Ineq.~\eqref{eq:cardinality} ensures
that at most $k$ new sites are selected in the answer set. Since $Z$
monotonically decreases with $\sum_{i=1}^n x_i$, $Z$ will attain its optimal
value only when $\sum_{i=1}^n x_i = k+|\fac|$.  The constraint in
Ineq.~\eqref{eq:trajclus1} guarantees that each trajectory is served by at most
one service location. The constraint in Ineq.~\eqref{eq:total_traj_served}
ensures that at least $\ufs . |\traj|$ trajectories are served.  The constraint
in Ineq.~\eqref{eq:trajclus2} guarantees that if $x_i=0$, then $\forall j,\
y_{ij}=0$, i.e., no trajectory is served by the site $s_i$ as it is not a
serving facility location.

The optimal algorithm is impractical except for very small datasets
(demonstrated in Sec.~\ref{sec:maxtips_exp}).  Therefore, we next present a
couple of polynomial-time heuristics to solve \maxtips.

\subsection{\mif Algorithm}
\label{sec:mif}

\maxtips problem being a generalization of the $k$-center problem
(Sec.~\ref{sec:related}), the most natural approach to solve \maxtips is to
extend the greedy heuristic for the $k$-center problem
\cite{gonzalez1985clustering}. We refer to this adaptation as  \emph{\mif
(Most-Inconvenient-First)}.   It iteratively selects a site and the
corresponding most inconvenient trajectory in each of the $k$ iterations. The
details are as follows. 

Initially, we set $ \q=\fac$, i.e., the existing set of facilities. The
algorithm maintains a map, called the \emph{Nearest Facility} map, denoted by
\nn. This map keeps the trajectories in $\traj$ in a sorted order based on
their distance to the nearest facility in \q.  Let $\reptraj$ be an empty set
of representative trajectories. The algorithm runs in iterations. At the
beginning of each iteration, a trajectory $T_i \in \nn$ is chosen whose rank is
$\lfloor \trajthres\rfloor$ in the sorted ordering. The reason behind this
choice is that $T_i$ faces the maximum inconvenience among the first $\lfloor
\trajthres\rfloor$ trajectories in $\nn$ in the sorted ordering. $T_i$ is then
added to the set \reptraj. In case of no existing facilities, \nn is initially
empty. In such a scenario, any random trajectory is added to \reptraj in the
first iteration. Next, we choose a candidate site $s_i \in \s \setminus \q$
that is nearest to $T_i$, and add it to the set \q. If there are multiple such
candidate sites, then the tie is broken arbitrarily.    The above process is
repeated until $k$ new facility locations are selected.

Let us evaluate this algorithm on the example in Fig.~\ref{fig:example} with
$k=2$ and $\ufs=1$. Initially, \q is set to $\{s_0\}$. Since $T_5$ is the
farthest trajectory from $s_0$, it is added to \reptraj. Next, $s_3$ is chosen
and added to \q, because it is the nearest site to $T_5$. Subsequently, in the
next iteration, $T_6$ is added to \reptraj. As a result, $s_4$ is added to the
answer set. Finally, the algorithm concludes with the selection
$\q=\{s_0,s_3,s_4\}$, which is also the optimal answer.

Now, consider the same example with $k=2$ and $\ufs=0.8$. Once again, \q is
initialized with $s_0$. In iteration 1,  $T_4$ and $s_3$ are chosen, and in the
next iteration, $T_3$ and $s_2$ are chosen. The final selection is thus,
$\q=\{s_0,s_2,s_3\}$, which is again the optimal solution.


We observe that for any set of  sites \q, and two sets of trajectories,
$\subtraj,\traj$ such that $\subtraj \subseteq \traj$, the maximum
inconvenience faced by any trajectory in \subtraj due to the set \q is at most
the maximum inconvenience faced by any trajectory in \traj due to the set \q.
Therefore, any approximation bound that holds for the \maxtips problem with \uf
$\ufs=1$ will also hold for $\ufs <1$. Hence, we next discuss the approximation
results only for $\ufs=1$. Further, for ease of analysis, we assume that all
the nodes in the road network are candidate sites, i.e., $\s=V$.

\begin{thm}
	\label{thm:mif approx bound}
	Let $d$ and $d^*$ be the maximum inconvenience offered by the answer sets
	returned by the \mif algorithm, and the optimal algorithm for \maxtips,
	respectively. Then, $d \le 2d^*$ for $k=1$, and $d \le 2d^*+L$ for $k \ge
	2$ where $L$ is the length of the longest trajectory.
\end{thm}

The proof is given in Appendix~\ref{sec:mif approx bound}.

\begin{thm}
	\label{thm:mif complexity}
	The time and space complexities of \mif algorithm are $O(k.l.n\log n +
	k.m.l^2 + k.m\log m)$ and  $O(l(n+m))$ respectively, where $n=|V|$ is the
	total number of nodes in the road network, $l$ is the maximum number of
	nodes in any trajectory, $m$ is the total number of trajectories in \traj,
	and $k$ is the total number of iterations.
\end{thm}

The proof is stated in Appendix~\ref{sec:mif complexity proof}.

Although \mif offers bounded quality guarantees, it is quite slow. This is
because it does not use any pre-computed distances. The next scheme, however,
leverages on pre-computed distances, and offers significantly faster response
times.

\subsection{Algorithm using \nc}
\label{sec:nc}

We observe that  \mif takes significant computation time for calculating
 node-to-node distances, and node-to-trajectory distances. This is
because we cannot afford to pre-compute and store all pairs node-to-trajectory
distances, which is overwhelmingly large. However, an indexing scheme can be
used that pre-computes and stores only a \emph{small} set of node-to trajectory
distances.

In this section, we propose a heuristic that uses the \nc indexing framework
\cite{tops_paper} that was originally designed to solve the following
\emph{\tops} problem \cite{tops_paper}: \\
\emph{Given a set of trajectories \traj, a set of candidate sites \s that can
host the services, the \tops problem with query parameters $(k,\tau)$ seeks
to report the best $k$ sites, $\mathcal{Q} \subseteq \mathcal{S}, \
|\mathcal{Q}| = k$, that \emph{cover} maximum number of trajectories.  It
is assumed that a site $s_i$ covers a trajectory $T_j$, if and only if
$\dr(T_j,s_i) \le \tau$, where $\tau$ is referred to as the coverage
threshold.}

\nc performs multi-resolution clustering of the nodes in the road network, $V$.
\nc maintains $t$ instances of index structures
$\mathcal{I}_0,\dots,\mathcal{I}_{t-1}$ of varying cluster radii. A particular
index instance is useful for a particular range of query coverage thresholds.
From one instance to the next, the radius is increased by a factor of $\epsinc$
for some $\eps > 0$.  Assume that the normal range of query coverage threshold
$\tau$ is $[\tau_{min}$, $\tau_{max})$. Then the total number of index
instances  is $t = \lfloor \log_{\epsinc} (\tau_{max} / \tau_{min}) \rfloor
+1$.  For each index instance,  \nc  maps the trajectories to the sequence of
clusters that they pass through.  

Intuitively, as the coverage threshold $\tau$  increases, the number of trajectories
covered by any set of candidate sites \q also increases. This was validated
empirically in \cite{tops_paper}. 
 
Exploiting the general monotonic behavior of the trajectory coverage with
respect to the coverage threshold $\tau$, we propose the following heuristic to
answer the \maxtips problem. Our goal is to  identify the smallest value of
$\tau$ such that there exists a set $\q \subseteq \s$ of size $k$  that covers
at least \trajthres number of trajectories in \traj. To guess this desired
value of $\tau$, we perform a binary search over the range of $\tau$, i.e.,
$[\tau_{min},\tau_{max}]$.

The algorithm proceeds in iterations. In each iteration, it computes the value
of the coverage threshold, $\tau=\frac{\tau^c_{min}+\tau^c_{max}}{2}$ where
$\tau^c_{min}$ and $\tau^c_{max}$ denote the current ranges.  Initially,
$\tau^c_{min} = \tau_{min}$ and $\tau^c_{max} = \tau_{max}$.  Next, the \tops
query with parameters $(k,\tau)$ is computed.  While doing so, the existing
facilities \fac must be taken into consideration \cite{tops_paper}. The
trajectories that lie within the coverage threshold $\tau$ of any existing
facility, are deemed to be \emph{covered}.  If the trajectory coverage value,
i.e., the number of trajectories covered by the set $\q \cup \fac$, is lower
than \trajthres, then $\tau^c_{min}$ is set to $\tau$, else $\tau^c_{max}$ is
set to $\tau$. Consequently, in the next iteration, \tops query is computed
with the revised value of $\tau$.  Since this process can continue forever, it
is stopped when the difference between $\tau^c_{max}$ and $\tau^c_{min}$ falls
below a desired precision.  Suppose the final iteration executed the \tops
query with parameters $(k,\tau^\prime)$ and returned the set \q. Then the
answer to the \maxtips problem is also \q with maximum inconvenience as
$\tau^\prime$.  We call this algorithm simply \nc.

As the monotonicity of the trajectory coverage w.r.t. the coverage
threshold $\tau$ is not guaranteed theoretically, the quality of this heuristic
cannot be bounded. Empirically, however, it performs the best in terms of both
running time and quality (Sec.~\ref{sec:exp}).

\begin{thm}
	\label{thm:netclus complexity}
	The time and space complexities of \nc algorithm are
	$O(\log_2(\tau_{max}/\tau_{min}). t_{\tops})$ and $O(t.(n+m.l))$
	respectively, where $O(t_{\tops})$ is the time required to answer a \tops
	query by \nc and $\tau_{min}$ and $\tau_{max}$ are the ranges of $\tau$
	values indexed by \nc. Further, $t=1+\lfloor
	\log_{1+\eps}(\tau_{max}/\tau_{min})\rfloor$ is the number of index
	instances, $\eps>0$  is the index resolution parameter, $m=|\traj|, n=|V|$
	and $l$ is the maximum number of nodes in any trajectory.
\end{thm}

The proof is given in Appendix~\ref{sec:netclus complexity proof}.

While the \nc approach is index-based, \mif is non-index based. In case of \nc,
the distance of the trajectories and sites to their respective cluster centres
is pre-computed, thereby making it efficient. For \mif, all the necessary
distance computations are performed online and, hence, it is slower.

%% file: meantips.tex
\section{Algorithms for \meantips}
\label{sec:meantips}

This section presents an optimal algorithm and two heuristics for the \meantips
problem.  

\subsection{Optimal Algorithm}
\label{sec:opti-mean}

The following \emph{integer linear program} solves the \meantips problem. As in
Sec.~\ref{sec:opti-max}, we assume that the set of candidate sites $\s=\{s_i|1
\le i \le n\}$ is augmented with the set of existing facilities $\fac$.
\begin{small}
\begin{align}
	\text{minimize } Z = \sum_{j=1}^m \sum_{i=1}^n \left[ \dr(T_j,s_i) \times y_{ij} \right] &\quad \text{ such that } \label{eq:obj2} \\	
	\sum_{i=1}^n x_i &\leq k + |\fac| \label{eq:cardinality2} \\
	\forall 1 \le j \le m,\ \sum_{i=1}^n y_{ij} &= 1  \label{eq:clus1} \\
	\forall 1 \le i \le n,\ \forall 1 \le j \le m,\ y_{ij} &\le x_i  \label{eq:clus2} \\		
	\forall 1 \le i \le n, \ 
	x_i &\in \{0,1\} \label{eq:binary3}\\
	\forall s_i \in \fac, \ x_i &=1 \label{eq:existing meantips} \\
		\forall 1 \le i \le n,\ \forall 1 \le j \le m \ 
	y_{ij} &\in \{0,1\}
	\label{eq:binary4}	
\end{align}
\end{small}

The objective function in Eq.~\eqref{eq:obj2} ensures that the total
(equivalently, mean) inconvenience is minimized over the set of selected sites.
The constraint in Ineq.~\eqref{eq:clus1} guarantees that each trajectory is
served by exactly one service location. The semantics of rest of the
constraints are same as those stated in Sec.~\ref{sec:opti-max}.

This optimal algorithm is impractical except for extremely small datasets
(shown in Sec.~\ref{sec:meantips_exp}).  Therefore, we next present the
following heuristics for the \meantips problem.

\subsection{\hcc Algorithm} 
\label{sec:hcc}

Recall that the \meantips problem generalizes the $k$-medoids problem
(Sec.~\ref{sec:related}).  Our first heuristic, \hcc, therefore, builds on the
three popular approaches for the $k$-medoids problem, PAM
\cite{rousseeuw1990finding}, CLARA \cite{rousseeuw1990finding}, and CLARANS
\cite{ng2002clarans}.  We first describe these approaches, and then discuss the
proposed \hcc algorithm.

\noindent
\textbf{PAM:} The basic idea of PAM is as follows. Given a set of $n$ objects,
it starts by choosing $k$ random objects, called as \emph{medoids}. Each
non-medoid object is assigned to its nearest medoid. The cost of a particular
clustering is the sum of the distances of each non-medoid object to its nearest
medoid. The PAM algorithm proceeds in iterations. In each iteration, it swaps
one of the existing medoids with a non-medoid object such that the cost of the
resulting clustering decreases. To realize the swap with minimal cost, it
computes the cost of each possible swap which are as many as $k(n-k)$. This step
is computationally  very expensive.  PAM terminates when it reaches a local
minima, i.e., there is no possible swap resulting in a solution with a lower
cost.

\noindent
\textbf{CLARA:} For large datasets, the PAM algorithm is infeasible owing to
its high running time. To address this limitation, the CLARA algorithm
\cite{rousseeuw1990finding} that relies on sampling was proposed. The idea is
to create random samples of size much smaller than $n$, and execute the PAM
algorithm on each of these samples, and return the clustering that offers the
minimal cost over all the samples. 

\noindent
\textbf{CLARANS:} To improve the quality of clustering of CLARA, another
algorithm called CLARANS was proposed \cite{ng2002clarans}. The basic idea of
CLARANS is to avoid scanning all possible swaps in each iteration of PAM.
Essentially, a small fraction of the total $k(n-k)$ possible swaps are scanned
and the swap that offers the minimal clustering cost is executed. To increase
the robustness, this scheme is repeated a few number of times, and finally the
clustering with the minimal cost is reported.

Inspired by the above approaches, we propose a new algorithm, \emph{\hcc
(Hybrid-CLARA-CLARANS)}, for the \meantips problem that combines the ideas of
sampling (from CLARA) and scanning a small number of swaps (from CLARANS). The
basic idea is to consider a few samples of sufficiently small size, and for
each sample, to scan a sufficiently small number of swaps. The details are
described next.

\subsubsection{Details}
 
First, we describe the basic algorithm, and later, discuss how to make it
scalable.  Given a set \traj of $m$ trajectories, and a set \s of $n$ candidate
sites, initially, we compute the distances between each pair of trajectory and
site. Then we execute the following steps. 
 
Initialize \q to \fac, i.e., the set of existing facilities.  Choose a random
set of $k$ sites in $\s$, referred to as \emph{medoids}, and add it to \q. The
total inconvenience of the set \q is $\totali(\q) = \sum_{T_j \in \traj} I_j$,
where $I_j = \min_{s_i \in \q}\dr(T_j,s_i)$.  To efficiently compute the value
of $\totali(\q)$, we use the \nn map (discussed in Sec.~\ref{sec:mif}) which
tracks the trajectories based on the distance to their nearest facility in \q.

Subsequently, the algorithm proceeds in iterations. In each iteration, it scans
a fraction of the total possible swaps between a medoid and a non-medoid site in
$\s \setminus \q$, as discussed above, and executes the swap that results in the
lowest total inconvenience $\totali(\q)$. Out of the total number of possible
swaps, $k(n-k)$, \hcc scans only a fraction, referred to as the \emph{\nf}, and
denoted by $\nfs$. This process terminates when there is no possible swap that
leads to a solution with a lower total inconvenience, or after a pre-specified
number of maximum iterations $\eta$, whatever is encountered earlier. While
performing the swaps, we ensure that we do not swap any existing facility in \q.
The algorithm makes $t \ (t \ge 1)$ trials of the above steps to account for the
randomly chosen initial set of medoids.  Finally, it returns the set $\q$ with
the lowest total inconvenience achieved over the $t$ trials. 

Let us see the working of this algorithm for the \meantips problem on the
example shown in Fig.~\ref{fig:example} with $k=2$. We will consider a single
trial with \nf  $\nfs=1$. Initially, $\q$ is set to $\{s_0\}$. Next, let us
start with the following set of initial medoids $\{s_1,s_2\}$ that is added to
\q. We note that $\totali(\q)=42$. Since $\nfs=1$, we examine all possible swaps
resulting in the following sets: $\{s_0,s_1,s_3\}, \{s_0,s_1,s_4\},
\{s_0,s_2,s_3\}$ and $\{s_0,s_2,s_4\}$. Since  $\totali$ $(\{s_0,s_2,s_3\})=21$
is minimal, it executes this swap. Since there is no further possible swap that
results in lower total inconvenience, the algorithm terminates, reporting the
set $\{s_2,s_3\}$ as the answer. Incidentally, this happens to be the optimal
solution.

This heuristic requires the distance values between each pair of trajectory $T_j
\in \traj$ and site  $s_i \in \s$. Computing and storing these distances for
large datasets may be infeasible when $|\s|$ or $|\traj|$ is large.  Thus, under
such circumstances, we propose to  sample the set of candidate sites and
trajectories, using the following schemes. 

\subsubsection{Site Sampling}
\label{sec:site sampling}

The sampling technique is based on a simple clustering idea that clusters the
set of nodes $V$ in the road network and samples at most a single candidate site
from each cluster. The details are as follows.  A random node  $v_p \in V$ (that
is not yet clustered) is chosen as the cluster-center of a cluster $\clus(v_p)$
that consists of all nodes $v_q \in V$ that are not yet clustered and are
within some distance threshold $R$ from $v_p$. This process is repeated until
every node in $V$ is clustered. Finally, a sample $\subs \subseteq \s $ is
created by selecting at most a single site from each cluster that is closest to
the cluster-center, which we refer to as its cluster-representative. If a
cluster has no candidate site, then it  has no cluster-representative either.
As the chosen sites belong to different clusters, they are not expected to be
close to each other. Thus, the sampled sites are typically nicely distributed
over the road network. 

\subsubsection{Trajectory Sampling}
\label{sec:trajectory sampling}

Consider any trajectory $T_j \in \traj$. For each node $v_p \in T_j$, let
$v_p^\prime$ be the cluster-center of the cluster that contains $v_p$. Each
trajectory $T_j$ is mapped to a set of cluster-centers
$T_j^\prime=\{v_p^\prime\}$. This transforms the trajectories into a coarser
representation as nodes that are close to each other in the trajectory are
likely to fall in the same cluster. This transformed set of trajectories
$\{T_j^\prime\}$ is denoted by $\subtraj$.

Following this, the set of trajectories \subtraj is clustered based on their
Jaccard similarity measure, as proposed in \cite{harshbhandari2016clustering}.
The high level overview of this method is as follows.  Suppose, we are required
to create a trajectory sample of size $s$. Initially, each trajectory
$T_j^\prime \in \subtraj$ is a cluster by itself with $T_j^\prime$ being the
cluster-representative. For any two trajectories, $T_p^\prime$ and
$T_q^\prime$, their Jaccard similarity is given by
$J(T_p^\prime,T_q^\prime) = {|T_p^\prime\cap T_q^\prime|} / {|T_p^\prime \cup T_q^\prime|}$.

The clustering follows an iterative algorithm where in each iteration it fuses a
pair of clusters with cluster-representatives $T_p^\prime$ and $T_q^\prime$ that
have the maximum Jaccard similarity. After fusing, either of the two
trajectories, $T_p^\prime$ or $T_q^\prime$ is deemed as the
cluster-representative of the new cluster. The algorithm continues fusing a pair
of clusters in each iteration in this manner, until there are exactly $s$
clusters. The cluster-representatives of these $s$ clusters are mapped back to
their original representation as sequence of nodes, and reported as the desired
trajectory sample.

\begin{thm}
	\label{thm:hcc complexity}
	The time and space complexities of \hcc are
	$O(t.\eta.\nfs.k^2(n^\prime-k).m^\prime)$ and $O(m^\prime.(n^\prime + l))$
	respectively, where $n^\prime$ and $m^\prime$ are the number of sites and
	trajectories produced after sampling of the sites and the trajectories,
	respectively. Here, $k$ is the number of medoids, $\nfs$ is the
	swap-fraction, $t$ is the number of trials, $\eta$ is the maximum number of
	iterations per trial, and $l$ is the maximum number of nodes in any
	trajectory.
\end{thm}

The proof is available in Appendix~\ref{sec:hcc complexity}.

Unfortunately, \hcc does not have any approximation bound, since
the original algorithms, PAM, CLARA and CLARANS do not have any quality
guarantees either.

\subsection{GREAT Algorithm}
\label{sec:great}

We next propose a greedy heuristic called \incg (GREedy Avg-Tips) that offers
bounded quality guarantees.

It is an iterative algorithm that works on the principle of maximizing the
marginal gain in each successive iteration. It starts with the set $\q
\leftarrow \fac$. In each iteration $\theta=\{1,\dots,k\}$, it selects a site
$s_\theta \in \s \setminus \q$ such that the total inconvenience of the
resulting set,  $\totali(\q \cup \{s_\theta\})$, is minimized.  The site
$s_\theta$ is added to the set $\q$. The algorithm terminates after $k$
iterations.

Similar to \hcc, the \incg algorithm also assumes that the distances between
each pair of trajectory $T_j \in \traj$ and site $s_i \in \s$ are pre-computed
and available. We also use the \nn map (as in the case of \hcc) to track the
distance between each trajectory and its nearest facility. Whenever a new site
is chosen to be added into \q, we check whether it would be the nearest facility
for each trajectory $T_j$, and update it accordingly.

Let us see the working of this algorithm for the \meantips problem on the
example shown in Fig.~\ref{fig:example} with $k=2$. Before the algorithm begins,
we set $\q \leftarrow \{s_0\}$. In the first iteration, $s_3$ is selected as the
set $\{s_0,s_3\}$ offers the least total inconvenience of $43$ units. In the
next iteration, $s_2$ is  selected, resulting in optimal total inconvenience of
$21$ units.
 
\subsubsection{Analysis of \incg}

We, next, state an important property of \meantips problem, which in turn, helps us to bound the quality of \incg.

A function $f$ defined on any subset of a set $\mathcal{S}$ is
\emph{sub-modular} if for any pair of subsets $\mathcal{Q,R} \subseteq
\mathcal{S}$, $f(\mathcal{Q}) + f(\mathcal{R}) \geq f(\mathcal{Q} \cup
\mathcal{R}) + f(\mathcal{Q} \cap \mathcal{R})$ \cite{nemhauser1978analysis}. A
function $f$ is \emph{super-modular} if its negative $(-f)$ is sub-modular. The
following result shows that the function $\totali(\q)$ is super-modular.

\begin{thm}
	\label{thm:supermodular}
	For any set of candidate sites $\q \subseteq \s$, the total inconvenience
	$\totali(\q)$ is a non-increasing super-modular function. 
\end{thm}

\begin{proof}
	Consider any pair of sets $\mathcal{Q,R} \subseteq \mathcal{S}$ such that
	$\mathcal{Q} \subseteq \mathcal{R}$. 

	First, we show that $\totali(\q)$ is a non-increasing function.	Since
	$I_j(\q) = \min_{s_i \in \mathcal{Q}}\{\dr(T_j,s_i)\}$ is a minimum
	function over the set $\q$, it follows that $I_j(\mathcal{R})\leq
	I_j(\mathcal{Q})$. Thus,	$\totali(\mathcal{R}) = \sum_{j=1}^m
	I_j(\mathcal{R}) \leq \sum_{j=1}^m I_j(\mathcal{Q}) =
	\totali(\mathcal{Q})$. Hence, $\totali(\q)$ is a non-increasing function.
	
	To show that $\totali(\q)$ is super-modular,
	it is sufficient to show that 
	for any site $s \in \mathcal{S} \setminus \mathcal{R}$, the following holds
	\cite{nemhauser1978analysis}:
	\begin{align}
		\totali(\mathcal{Q}\cup \{s\}) - \totali(\mathcal{Q})
			\leq \totali(\mathcal{R}\cup \{s\}) - \totali(\mathcal{R})
	\end{align}
	Since $\totali(\q) = \sum_{j=1}^m I_j$, it is, therefore, enough to prove
	that for any trajectory $T_j \in \mathcal{T}$,
	\begin{align}
		\label{eq:supermodular}
		I_j(\mathcal{Q}\cup \{s\}) - I_j(\mathcal{Q})
			\leq I_j(\mathcal{R}\cup \{s\}) - I_j(\mathcal{R})
	\end{align}

	Suppose the site $s^* \in \mathcal{R} \cup \{s\}$ is the nearest facility to
	trajectory $T_j$ in the set $\mathcal{R} \cup \{s\}$.  There can be
	two cases: \\
	(a) $s^* = s$: $I_j(\mathcal{Q} \cup \{s\})=I_j(\{s\})=I_j(\mathcal{R} \cup
	\{s\})$. Further, since $I_j(\mathcal{Q})\geq I_j(\mathcal{R})$
	(using the non-increasing property), Ineq.~\eqref{eq:supermodular} follows. \\
	(b) $s^* \neq s, s^* \in \mathcal{R}$: Here, $I_j(\mathcal{R} \cup
	\{s\}) - I_j(\mathcal{R}) = 0$.  Using the non-increasing property of $I_j(\q)$,
	$I_j(\mathcal{Q} \cup \{s\}) -
	I_j(\mathcal{Q}) \leq 0$.  Thus, Ineq.~\eqref{eq:supermodular} follows.
	\hfill{}
\end{proof}

The next result bounds the quality of \incg.

\begin{thm}
	\label{thm:bound of incg for meantips}
	Let $OPT \subseteq \s, |OPT|=k$ denotes an optimal solution to the
	\meantips problem. Let $\q \subseteq \s, |\q|=k$ be the solution reported
	by the \incg algorithm. Then,
	\begin{align*}
		\totali(\q)&=\totali(OPT) & \mbox{ for } k=1 \\
		\totali(\q) &\le  (1-1/e)\totali(OPT)+ \totali(\fac)/e  & \mbox{
		for } k \ge 2
	\end{align*}
	where $\totali(\fac)$ refers to the initial total inconvenience offered by
	the existing facilities \fac. We assume $\totali(\fac)$ to be a
	non-negative constant.
\end{thm}

The proof is stated in Appendix~\ref{sec:great approx bound}.

The next result bounds the complexity of \incg.

\begin{thm}
	\label{thm:complexity of incg for meantips}
	The space and time complexities of \incg are $O(m.(n+l))$ and $O(k.m.n)$
	respectively, where $m=|\traj|$, $n=|\s|$ and $l$ is the maximum length of
	any trajectory.
\end{thm}

The proof is available in Appendix~\ref{sec:complexity of incg for meantips}.

The key drawback of the above scheme is its space requirement of $O(m.n)$ which
is prohibitively large for city-scale datasets. To alleviate this problem, we
work with sampled set of trajectories and candidate sites, as described in
Sec.~\ref{sec:site sampling} and Sec.~\ref{sec:trajectory sampling}.

%% file: results.tex
\renewcommand{\subsubsection}[1]{\noindent\textbf{#1:}}

\section{Experimental Evaluation}
\label{sec:exp}

\setlength{\tabcolsep}{3pt}	

\begin{table}[t]\small
\centering
\begin{tabular}{llrr}
\toprule
\bf Dataset & \bf Type & \bf \# Trajectories & \bf \# Sites \\
\midrule
\bs (BS) & Real & 8,083 & 30 \\
\bm (BM) & Real & 8,083 & 47,125 \\
\bl (BL) & Real & 123,179 & 269,686 \\
\midrule
\bms & \multirow{2}{*}{Real} & \multirow{2}{*}{1,278} & \multirow{2}{*}{1,883} \\
(BMS) & & & \\
\bls & \multirow{2}{*}{Real} & \multirow{2}{*}{9,701} & \multirow{2}{*}{15,775} \\
(BLS) & & & \\
\midrule
Bengaluru & Synthetic & 9,950 & 61,563 \\
New York & Synthetic & 9,950 & 355,930 \\
Atlanta & Synthetic & 9,950 & 389,680 \\
\bottomrule
\end{tabular}
\tabcaption{Summary of datasets.}
\label{tab:datasets}
\end{table}

\setlength{\tabcolsep}{6pt}	

In this section, we perform extensive experiments to assess the quality,
scalability and practicality of the different heuristics.  Since both the \tips
problems, \maxtips and \meantips, are introduced in this work, there is
\emph{no} existing \emph{practical} baseline technique to compare against.
Nevertheless, we do compare with the nearest available baseline techniques for
both \emph{static} and \emph{mobile} users' scenarios.  Further, we compare with
the optimal algorithms for both \maxtips and \meantips, on small datasets.  The
experiments were conducted using Java (version 1.7.0) code on an Intel(R) Core
i7-4770 CPU @3.40GHz machine with 32GB RAM running Ubuntu 14.04.2 LTS OS.

\subsection{Datasets} 
\label{sec:datasets}
 
We conducted experiments on both real and synthetic datasets, whose details are
shown in Table~\ref{tab:datasets}.  For simplicity, we assume that the set
of candidate sites is same as the set of nodes in the road network, unless
otherwise stated.

\noindent
\textbf{Real datasets:} We use GPS traces from Beijing consisting of user
trajectories generated by tracking taxis for a week \cite{cab2,cab1}.  To
generate trajectories as sequences of road intersections, the raw GPS-traces
were map-matched~\cite{map1} to the Beijing road network extracted from
OpenStreetMap (\url{http://www.openstreetmap.org/}).  This dataset, referred to
as \bl (BL) here, is the largest and most widely used publicly available
trajectory datasets.

For \meantips, all the algorithms require distances between each pair of site
and trajectory. For a large dataset such as BL, computing and storing such
pairwise distances is infeasible. Therefore, for thorough evaluation, we use a
medium sized dataset, \bm (BM). 

To assess the quality of sampling techniques, representative datasets, referred
to as \bms (BMS) and \bls (BLS) respectively, are derived from the BM and BL
datasets.  These datasets are generated using the sampling schemes described in
Sec.~\ref{sec:site sampling} and Sec.~\ref{sec:trajectory sampling}.
 
Since both the \tips problems are NP-hard, the optimal algorithm requires
exponential time and, therefore, can be run only on a  small dataset.  Hence,
we evaluate all the algorithms against the optimal on \bs (BS) dataset which
consists of the same set of trajectories as in BM, but has only 30 candidate
sites, which are sampled randomly.  To increase the robustness of the results,
such sampling was performed 5 times, and the experiments were conducted 10
times for each sample.

\noindent
\textbf{Synthetic datasets:} To study the impact of city geographies, we
generated three synthetic datasets that emulate trajectory patterns
followed in New York, Atlanta and Bengaluru.  We use an online traffic
generator tool, MNTG (\url{http://mntg.cs.umn.edu/tg/index.php}) to generate
the traffic traces, and map-match them to generate trajectories in
the desired format.

\subsection{\maxtips Results}
\label{sec:maxtips_exp}

We evaluate the performance of three different algorithms to solve \maxtips,
\opt,  \mif, and \nc, on the two basic parameters: (i)~desired number of
service locations $k$, varied in the range $[1,10]$, and (ii)~\uf $\ufs$,
varied in the range $10\%$ to $100\%$.  The default values of $k$ and $\ufs$
are $5$ and $90\%$ respectively.  The metrics evaluated are (i)~maximum
inconvenience, $\maxi$, and (ii)~running time.  The dataset used is BL, unless
otherwise stated.

\begin{figure}[tb]
	\centering
	\moveup
	\moveup
	\subfloat[Max. Inconvenience.]
	{
		\includegraphics[width=\subfigwidth]{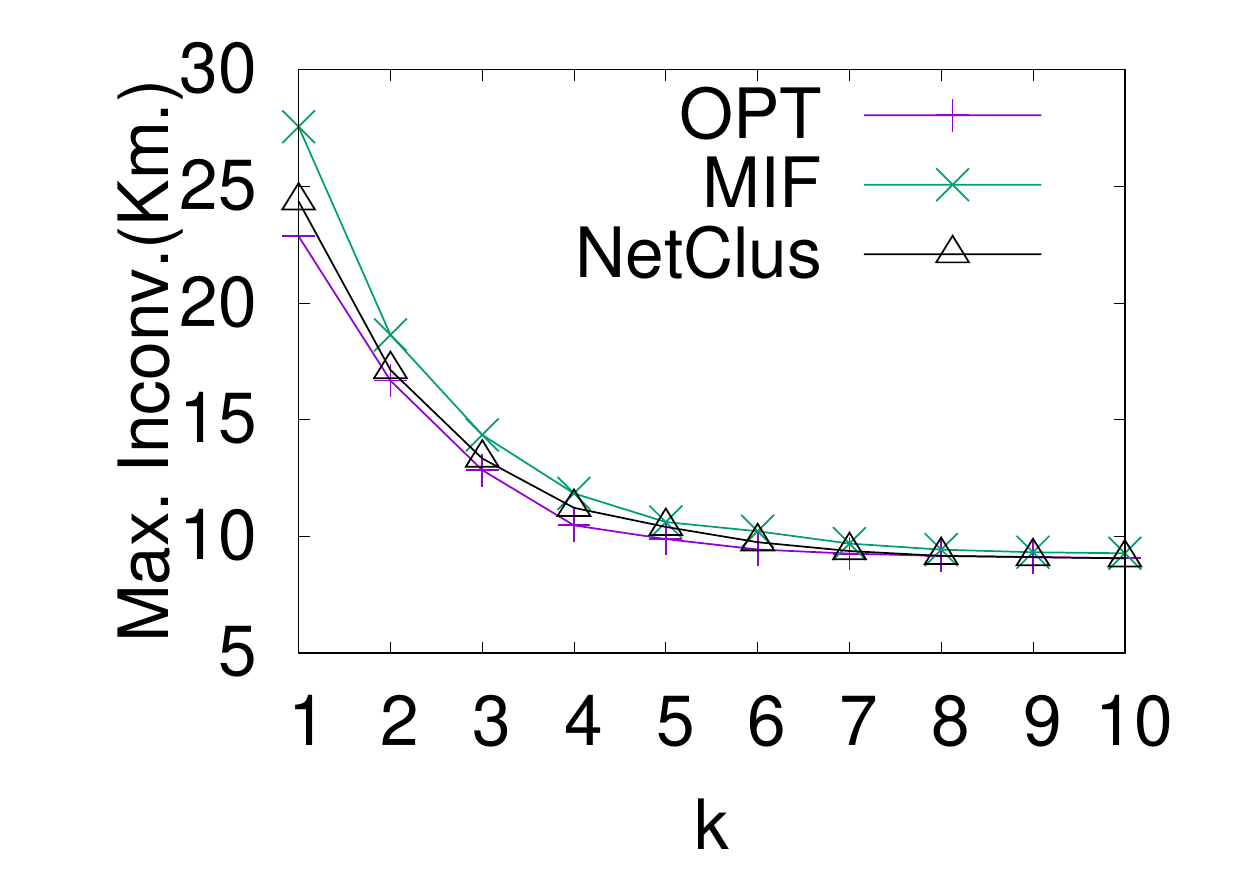}
		\label{subfig:maxtips optdist}
	}
	\subfloat[Running time.]
	{
		\includegraphics[width=\subfigwidth]{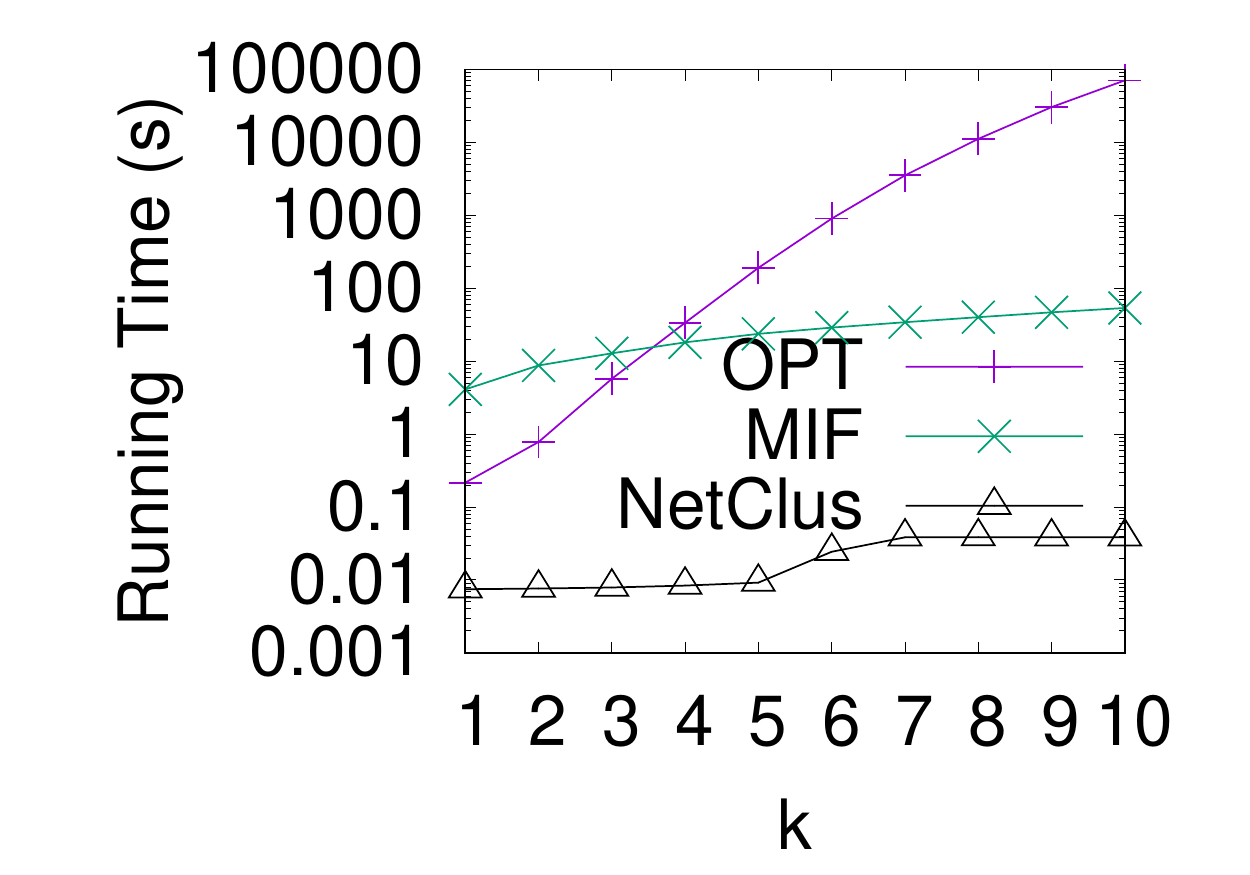}
		\label{subfig:maxtips opttime}
	}
	\figcaption{\maxtips: Comparison with optimal at $\ufs=100 \%$.}
	\label{fig:max_tips optimal}
\end{figure}

\subsubsection{Comparison with Optimal}
The optimal algorithm used, is the integer linear program (ILP), given in
Sec.~\ref{sec:maxtips}.  Since the optimal algorithm requires exponential
running times, we run it only on the BS dataset.  Fig.~\ref{fig:max_tips
optimal} shows that \maxi values offered by \mif and \nc are very close to that
of \opt although the running times are much better. \opt requires several hours
to complete even for this small dataset and, therefore, is not practical at
all.  Consequently, we do not experiment with \opt any further.

\begin{figure}[tb]
	\centering
	\moveup
	\moveup
	\subfloat[Varying $k$ at $\ufs=90\%$]
	{
		\includegraphics[width=\subfigwidth]{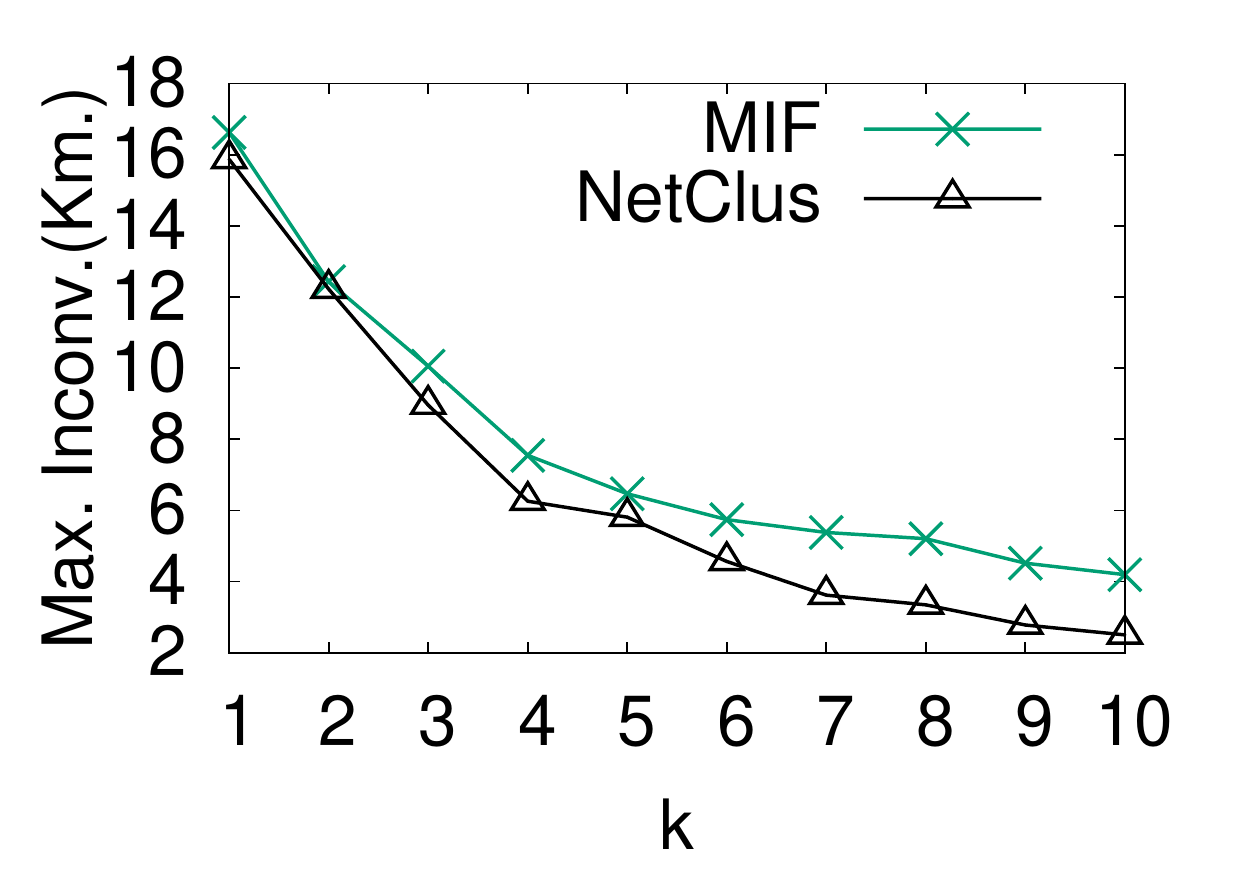}
		\label{subfig:maxtips k_dist}
	}
	\subfloat[Varying $\ufs$ at $k=5$]
	{
		\includegraphics[width=\subfigwidth]{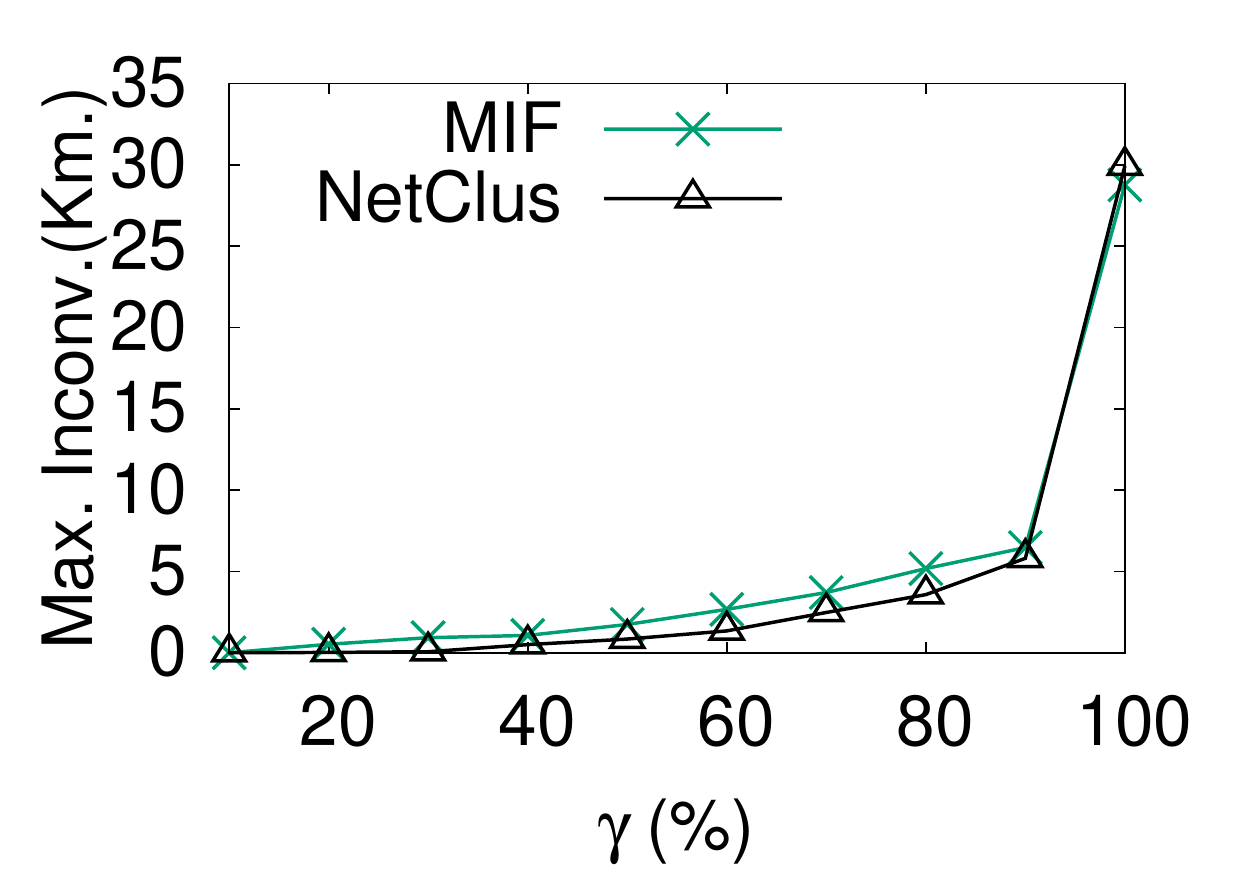}
		\label{subfig:maxtips uf_dist}
	}
	\figcaption{\maxtips: Quality results.}
	\label{fig:max_tips quality}
\end{figure}

\subsubsection{Quality Results}
Fig.~\ref{subfig:maxtips k_dist} shows the \maxi values of \mif and \nc  on the
BL dataset.  \nc offers the least \maxi value, beating  \mif by more than 20\%
on an average.

Fig.~\ref{subfig:maxtips uf_dist} shows that as \uf increases from 90\% to
100\%, there is a sharp rise in \maxi value for both \mif and \nc. This is
because there are generally trajectories that are very hard to satisfy and,
since at $\ufs=100\%$,  all of them need to be served, \maxi values shoot
up for both the algorithms.

\begin{figure}[tb]
	\centering
	\moveup
	\moveup
	\subfloat[Varying $k$ at $\ufs=90\%$]
	{
		\includegraphics[width=\subfigwidth]{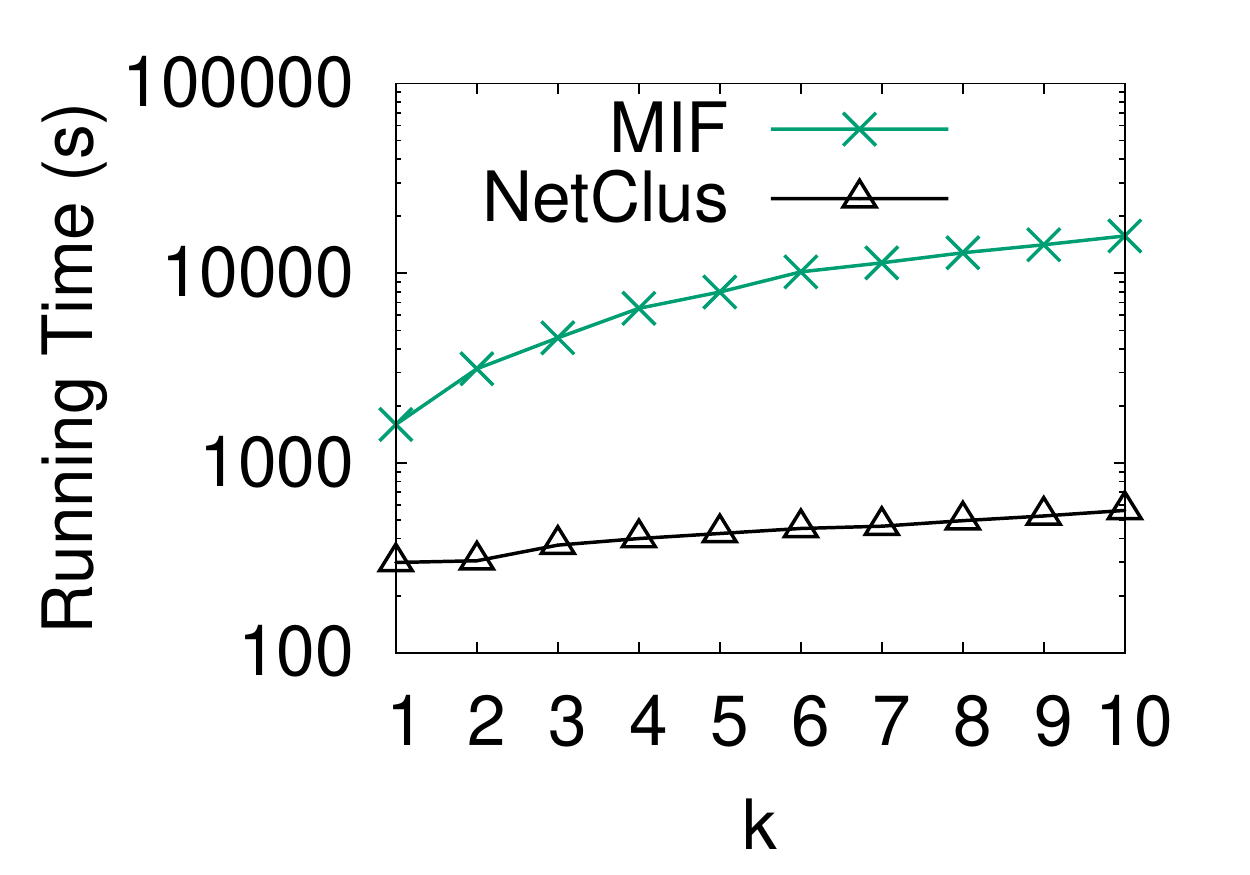}
		\label{subfig:maxtips ktime}
	}
	\subfloat[Varying $\ufs$ at $k=5$]
	{
		\includegraphics[width=\subfigwidth]{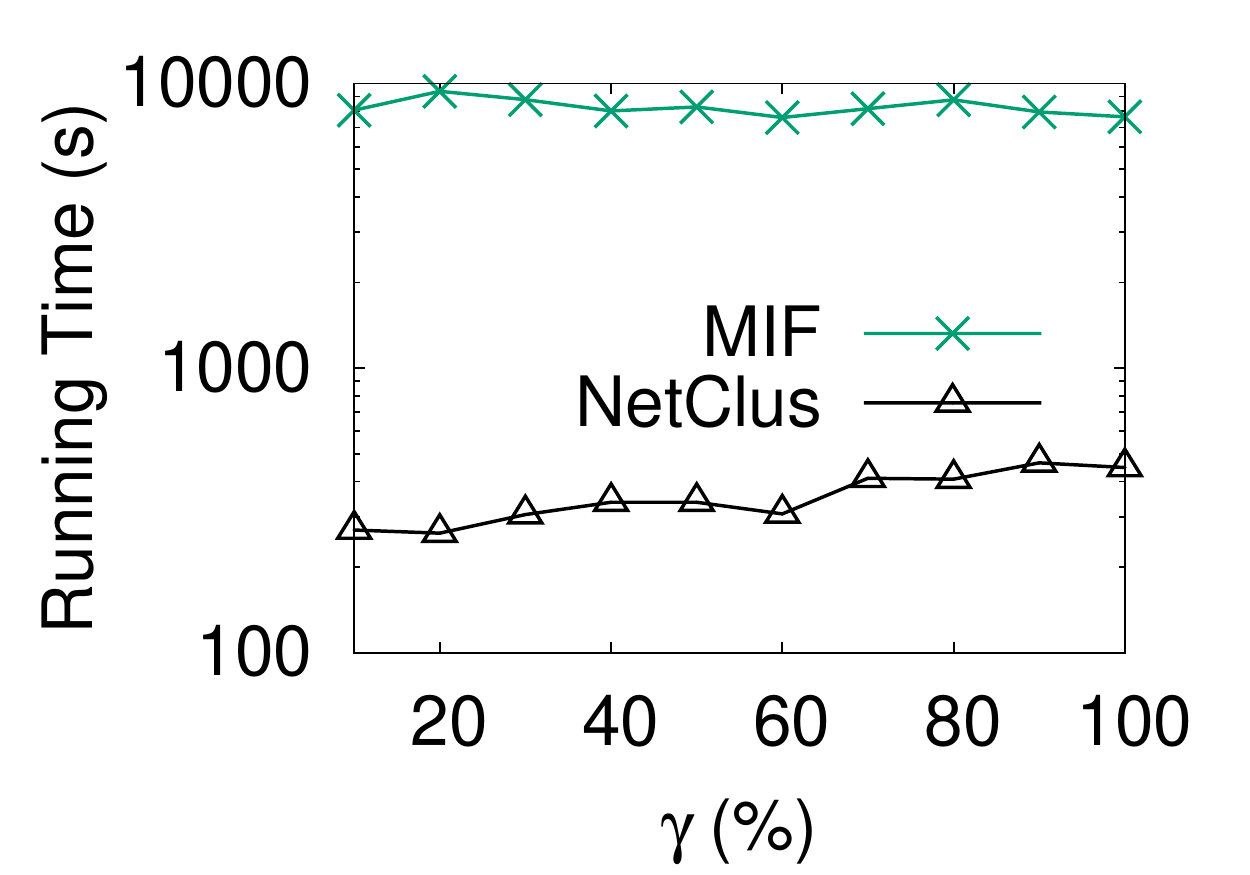}
		\label{subfig:maxtips uftime}
	}
	\figcaption{\maxtips: Running time performance.}
	\label{fig:max_tips time}
\end{figure}

\subsubsection{Performance Results}
The running time results portrayed in Fig.~\ref{fig:max_tips time} show that
\nc is 1-2 orders of magnitude faster than \mif. The high running time of \mif
is due to a large number of calls to the shortest path algorithm (once for each
site on the chosen trajectory in each of the $k$ iterations).  Moreover, to
guard against a poorly selected random initial trajectory, \mif is repeated
thrice.  On the other hand, \nc is fast since it uses pre-computed distances
between trajectories and centers of the clusters that they pass through.

\subsubsection{Memory Footprint}
The memory consumption of \mif and \nc on the BL dataset (at the default values
of $k$ and $\ufs$) are about 8.6GB and 3.2GB respectively. The low memory
footprint of \nc is due to clustering of the site space and consequent
compressed representation of trajectories.

\begin{figure}[tb]
	\centering
	\moveups
	\subfloat[Number of sites.]
	{
		\includegraphics[width=\subfigwidth]{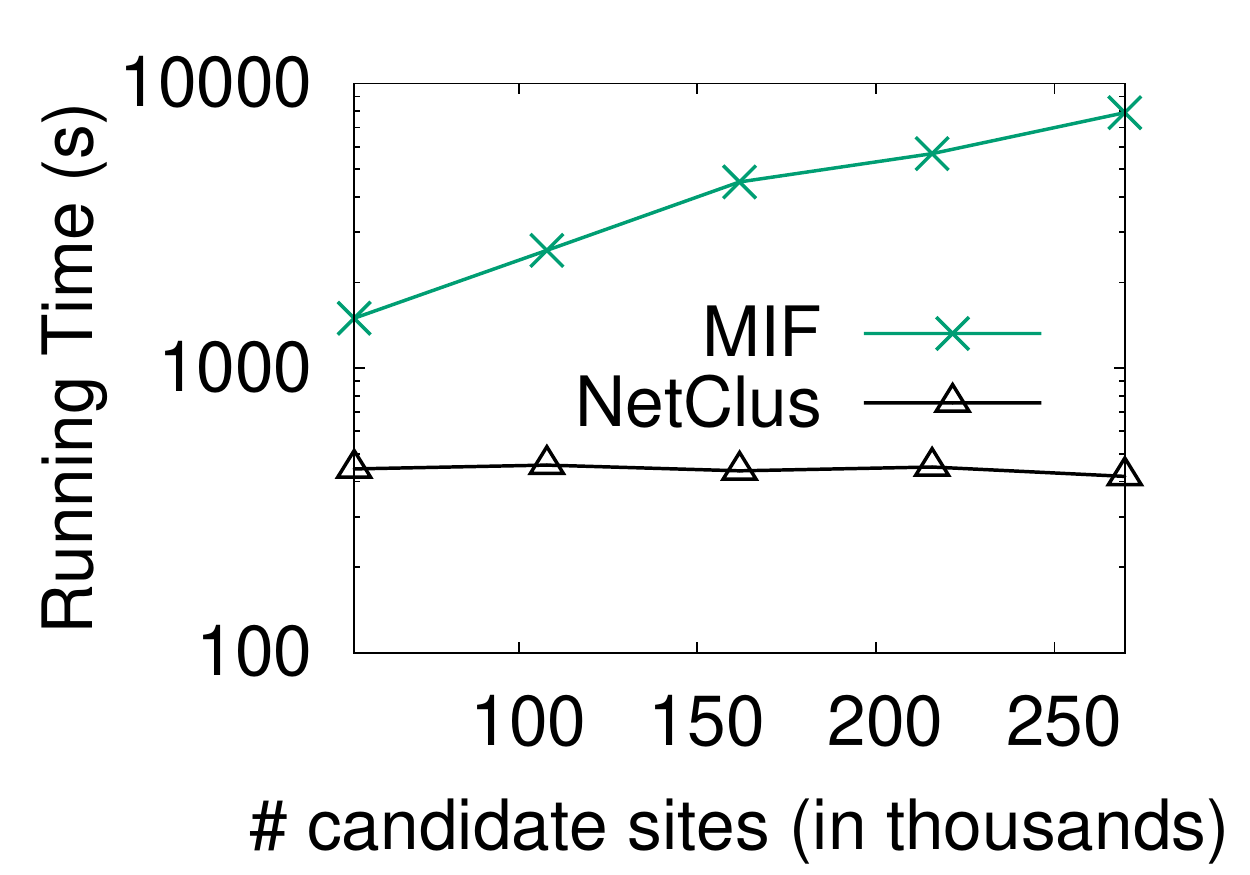}
		\label{subfig:maxtips site}
	}
	\subfloat[Number of trajectories.]
	{
		\includegraphics[width=\subfigwidth]{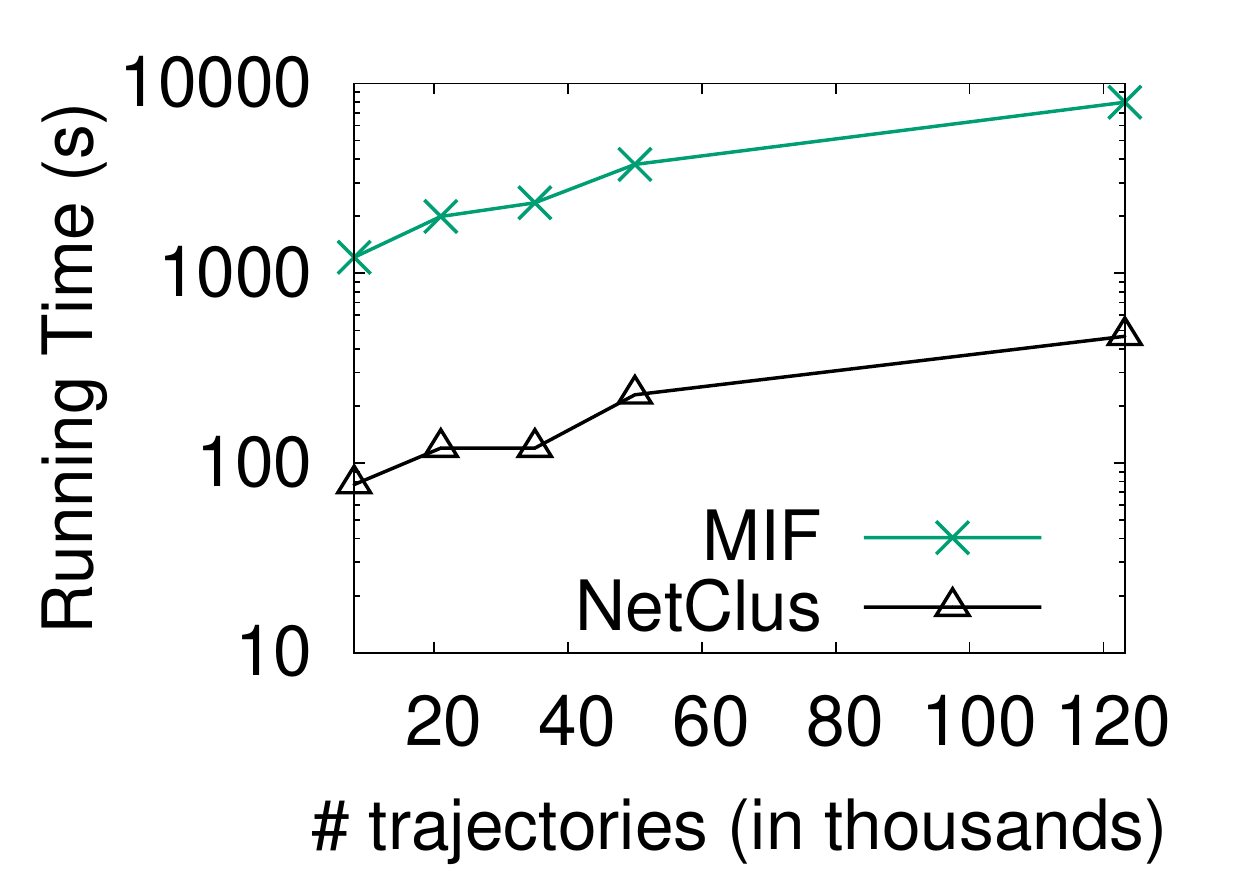}
		\label{subfig:maxtips traj}
	}
	\figcaption{\maxtips: Scalability results ($k=5$ and $\ufs=90\%$).}
	\label{fig:max_tips scale}
\end{figure}

\subsubsection{Scalability}
\label{sec:maxtips scalability}
To ascertain the scalability of \nc and \mif with respect to the number of
candidate sites and trajectories, we next took different sized samples from the
BL dataset.  Fig.~\ref{fig:max_tips scale} shows that \nc scales better with
the number of sites due to its clustered representation.  It is faster than
\mif by at least an order of magnitude for all the situations.

\begin{figure}[tb]
	\centering
	\moveups
	\subfloat[Max. Inconvenience.]
	{
		\includegraphics[width=\subfigwidth]{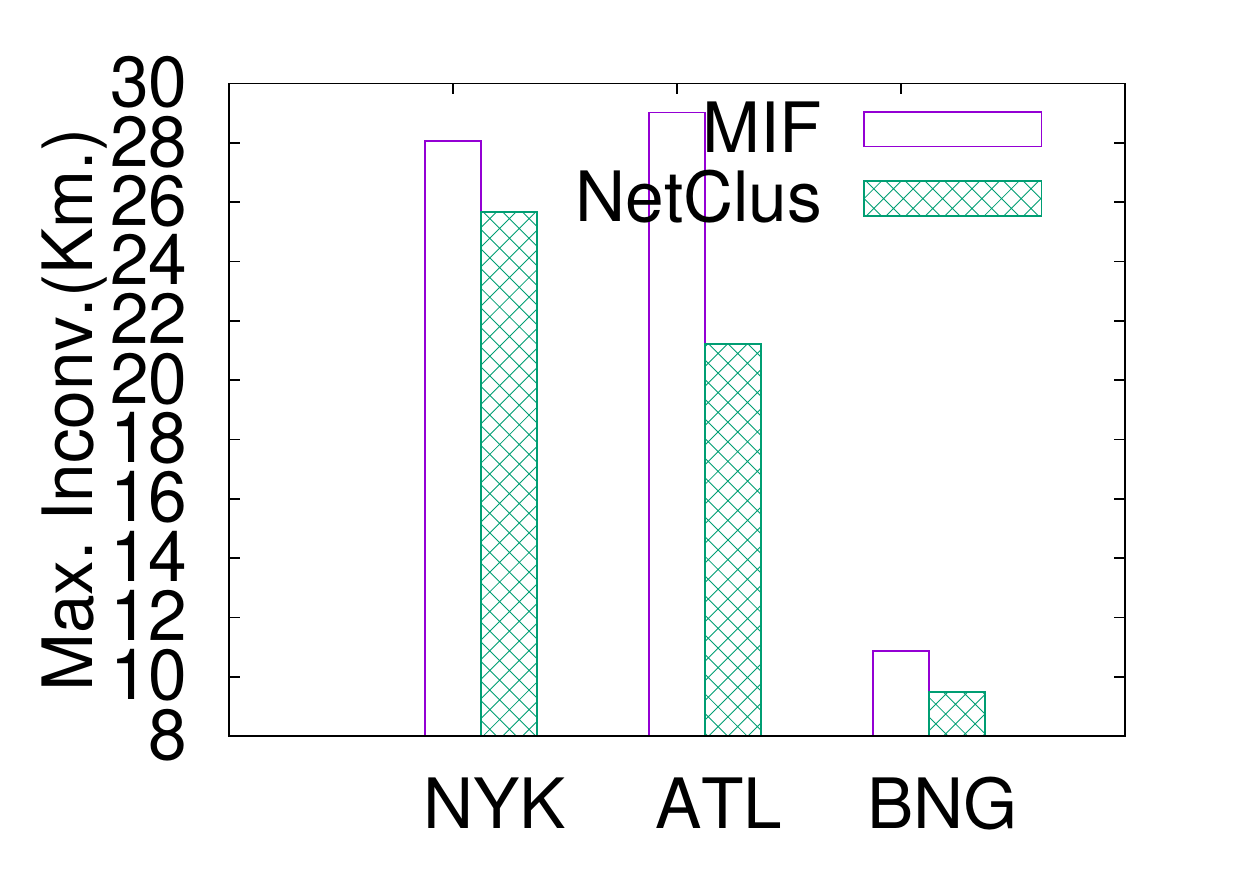}
		\label{subfig:maxtips syndist}
	}
	\subfloat[Running time.]
	{
		\includegraphics[width=\subfigwidth]{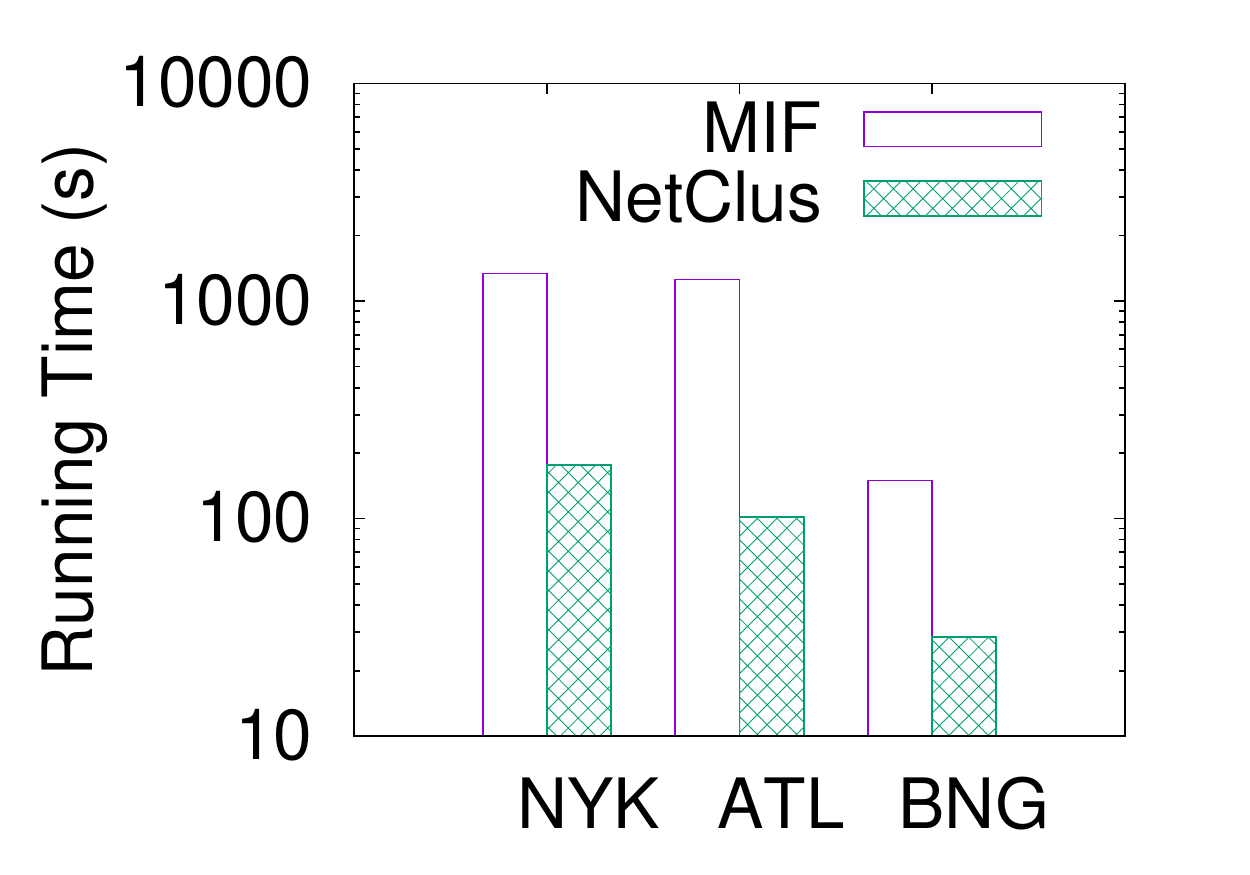}
		\label{subfig:maxtips syntime}
	}
	\figcaption{\maxtips: Synthetic datasets (at $k=5$ and $\ufs=90\%$).}
	\label{fig:max_tips synthetic}
\end{figure}

\subsubsection{Synthetic Datasets}
\label{sec:maxtips synthetic}
Fig.~\ref{fig:max_tips synthetic} shows the evaluation of \maxtips on three
synthetic datasets emulating traffic in Atlanta, New York, and Bengaluru.
Bengaluru has the smallest sized road network.  Therefore, it has the best
\maxi value and running time.  Atlanta and New York have much larger road
network, and the trajectories are distributed all over the network. Thus, they
are harder to be served, and consequently exhibit high \maxi values. Their
running times are high owing to large number of candidate sites to be
processed.

\subsection{\meantips Results}
\label{sec:meantips_exp}

We evaluated the performance of three different algorithms for \meantips, \opt,
\hcc, and \incg on the desired number of service
locations $k$, varied in the range $[1,10]$, with default as $5$. The metrics
evaluated are (i)~average inconvenience, $\ai(\q) = \totali(\q) / |\traj|$, and
(ii)~running time.

\begin{table}[b]\small
\centering
\begin{tabular}{rrr}
\toprule
\multirow{2}{*}{\bf $\nfs$ (in $\%$)} &	\multicolumn{2}{c}{\bf with respect to $\nfs=100\%$}	\\
	\cline{2-3}
& \bf Error & \bf Speed-up factor \\
\midrule
1	&	10.92	&	7.22	\\
2	&	 4.49	&	4.97	\\
5	&	 2.86	&	4.27	\\
10	&	 2.27	&	3.40	\\
25	&	 1.16	&	2.37	\\
50	&	 0.00	&	1.53	\\
\bottomrule
\end{tabular}
\tabcaption{Performance of \hcc with varying \nf.}
\label{tab:nf}
\end{table}

\subsubsection{Choice of \nf in \hcc algorithm}
Referring to Sec.~\ref{sec:hcc}, recall that \hcc scans only a fraction \nfs of
the total number of swaps, referred to as the \nf.  Table~\ref{tab:nf} shows the
performance of \hcc for different values of \nfs (shown as percentage of total
number of swaps). The second column lists the relative error in \ai value w.r.t.
$\nfs=100\%$, and the third column indicates the corresponding speed up in
running time. For our experimentation, we choose $\nfs=5\%$ as this offers a
nice balance with an error of less than 3\% and a speed up of more than 4 times.   

\begin{figure}[t]
	\centering
	\moveup
	\moveup
	\subfloat[Avg. Inconvenience.]
	{
		\includegraphics[width=\subfigwidth]{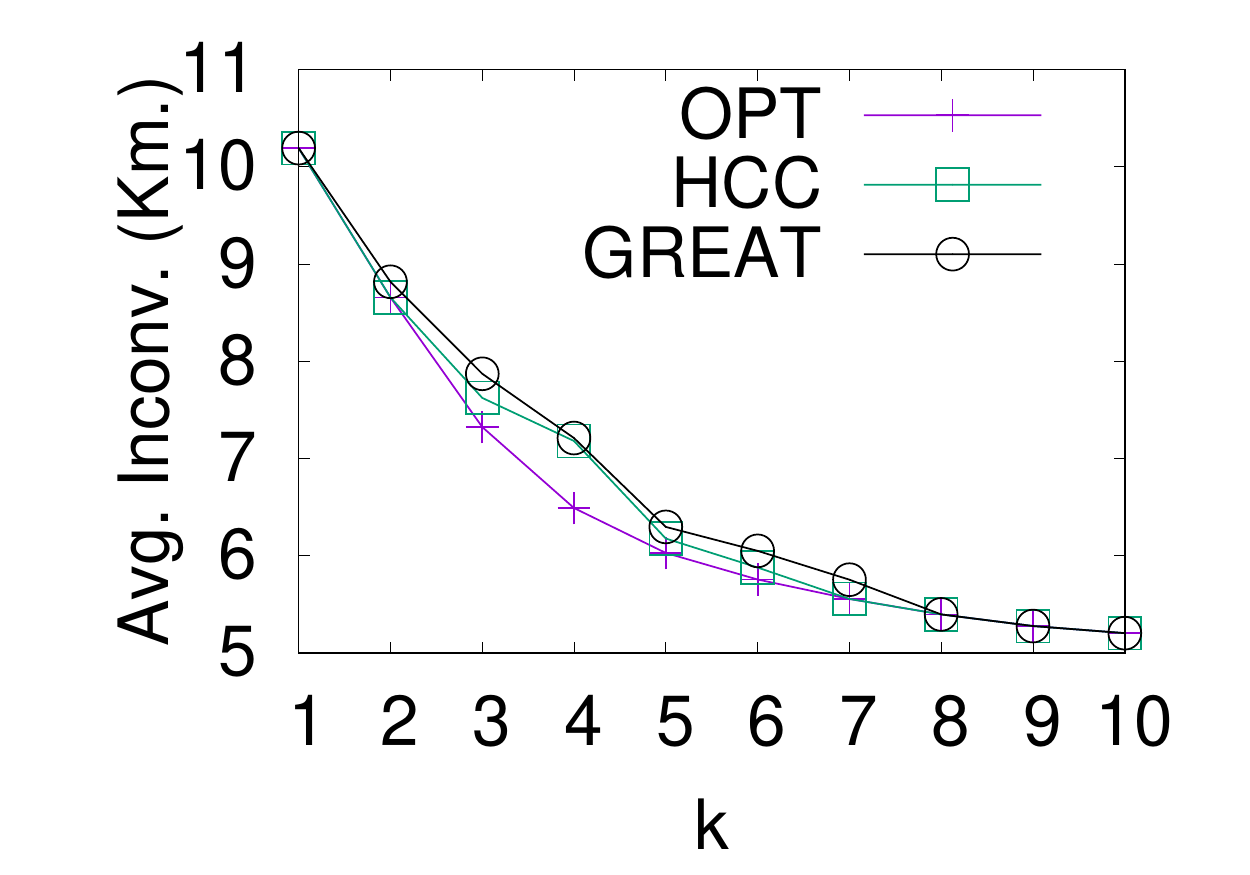}
		\label{subfig:meantips optdist}
	}
	\subfloat[Running time.]
	{
		\includegraphics[width=\subfigwidth]{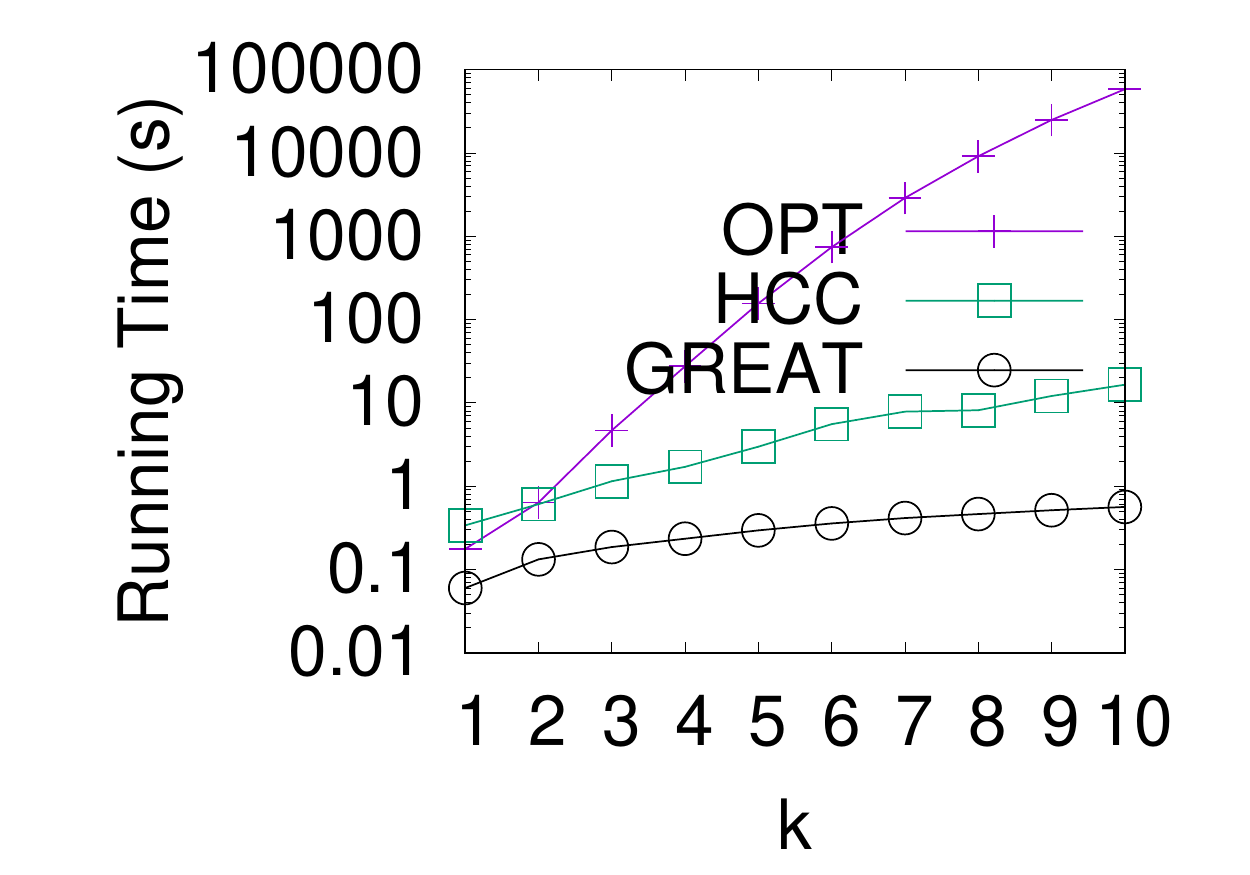}
		\label{subfig:meantips opttime}
	}
	\figcaption{\meantips: Comparison with optimal.}
	\label{fig:mean_tips optimal}
\end{figure}

\subsubsection{Comparison with Optimal}
Since \meantips is NP-Hard, the ILP-based optimal algorithm given in
Sec.~\ref{sec:meantips}, can be run only on small datasets. Therefore, as in the
case of \maxtips, we evaluate it on the \bs dataset.  Fig.~\ref{fig:mean_tips
optimal} shows that both \hcc and \inc offer almost the same quality as \opt.
However, \opt takes hours of query time even on such a small dataset.
Consequently, we drop \opt from further consideration.

\begin{figure}[tb]
	\centering
	\moveup
	\moveup
	\subfloat[Varying $k$ on BM, BMS]
	{
		\includegraphics[width=\subfigwidth]{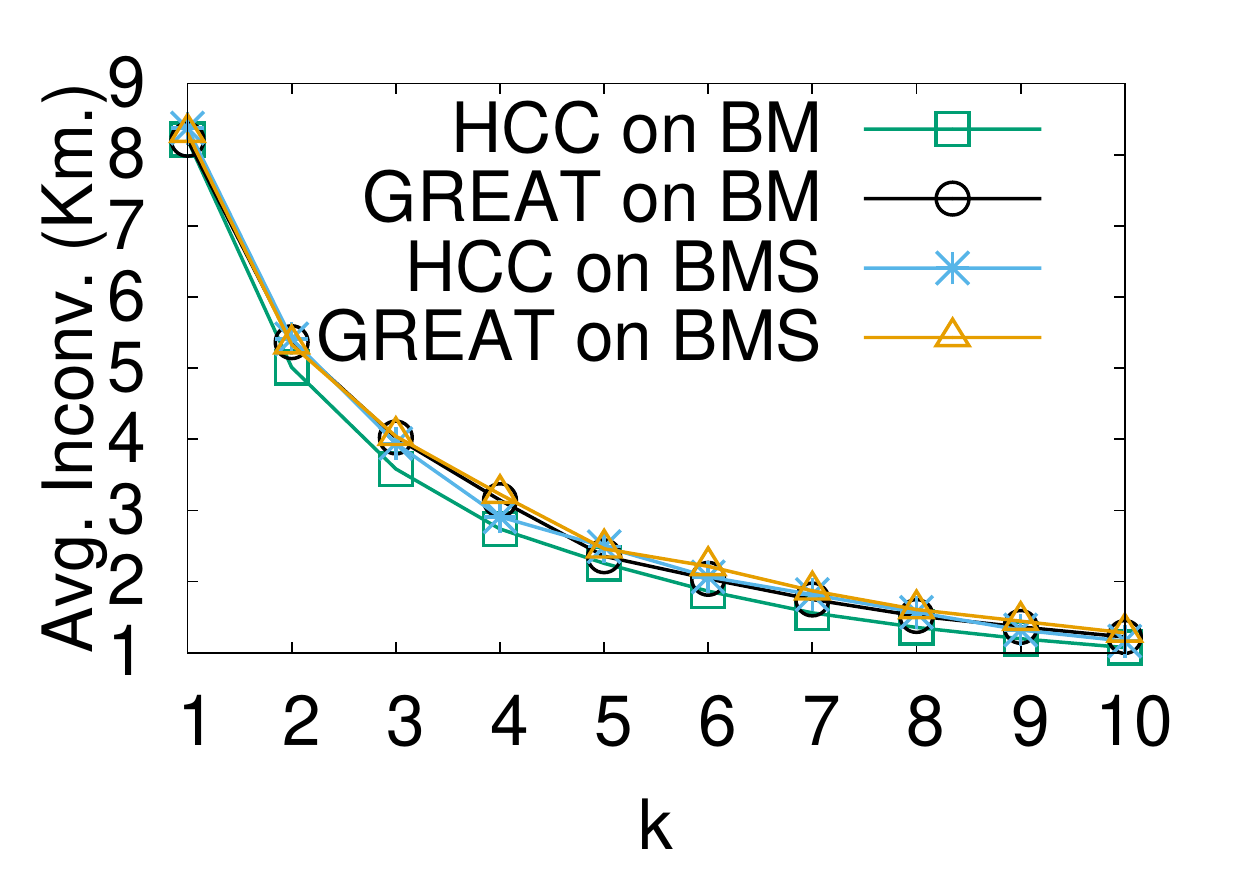}
		\label{subfig:meantips BM k_dist}
	}	
	\subfloat[Varying $k$ on BLS]
	{
		\includegraphics[width=\subfigwidth]{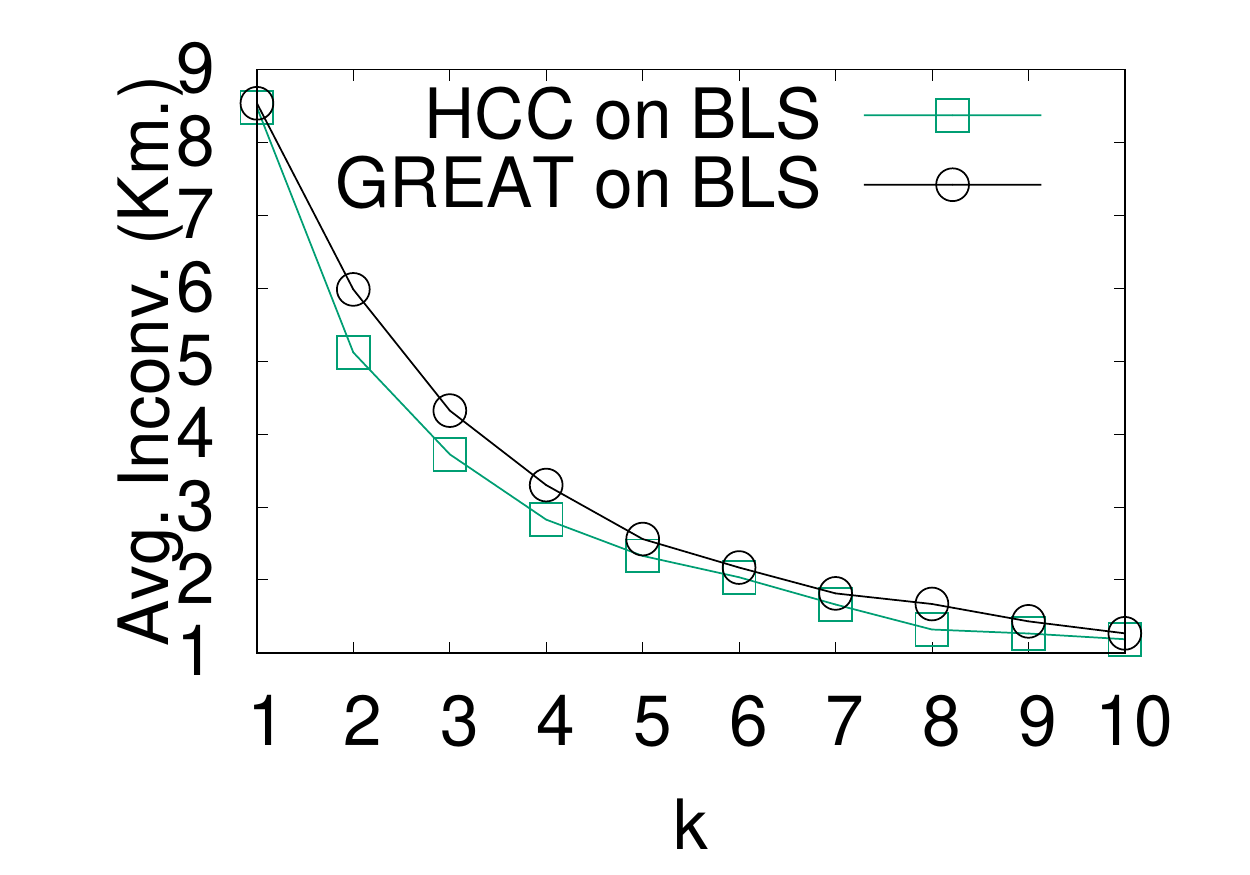}
		\label{subfig:meantips BL k_dist}
	}
	\figcaption{\meantips: Quality results.}
	\label{fig:mean_tips quality}
\end{figure}

\subsubsection{Quality Results}
As explained earlier, \hcc and \inc cannot be executed on the BL dataset due to
high memory overhead.  Hence, to ascertain the effect of site sampling and
trajectory sampling on the quality, we use \bm (BM), and its sampled
counterpart, the \bms (BMS) datasets.  Fig.~\ref{subfig:meantips BM k_dist}
shows that \ai values achieved by both \hcc and \inc on the sampled dataset BMS
is almost as good as the full BM dataset. The average relative error in \ai
values due to sampling is only about 10\% for \hcc and 4\% for \inc.
Fig.~\ref{subfig:meantips BL k_dist} reports the quality on the BLS dataset.
\hcc offers roughly 10\% better quality than \inc on an average.  Since the
error due to sampling  is fairly low, it is expected that the \ai values on the
BL dataset  are likely to be similar to those in the BLS dataset.

\begin{figure}[tb]
	\centering
	\moveup
	\moveup
	\subfloat[Varying $k$ on BM, BMS]
	{
		\includegraphics[width=\subfigwidth]{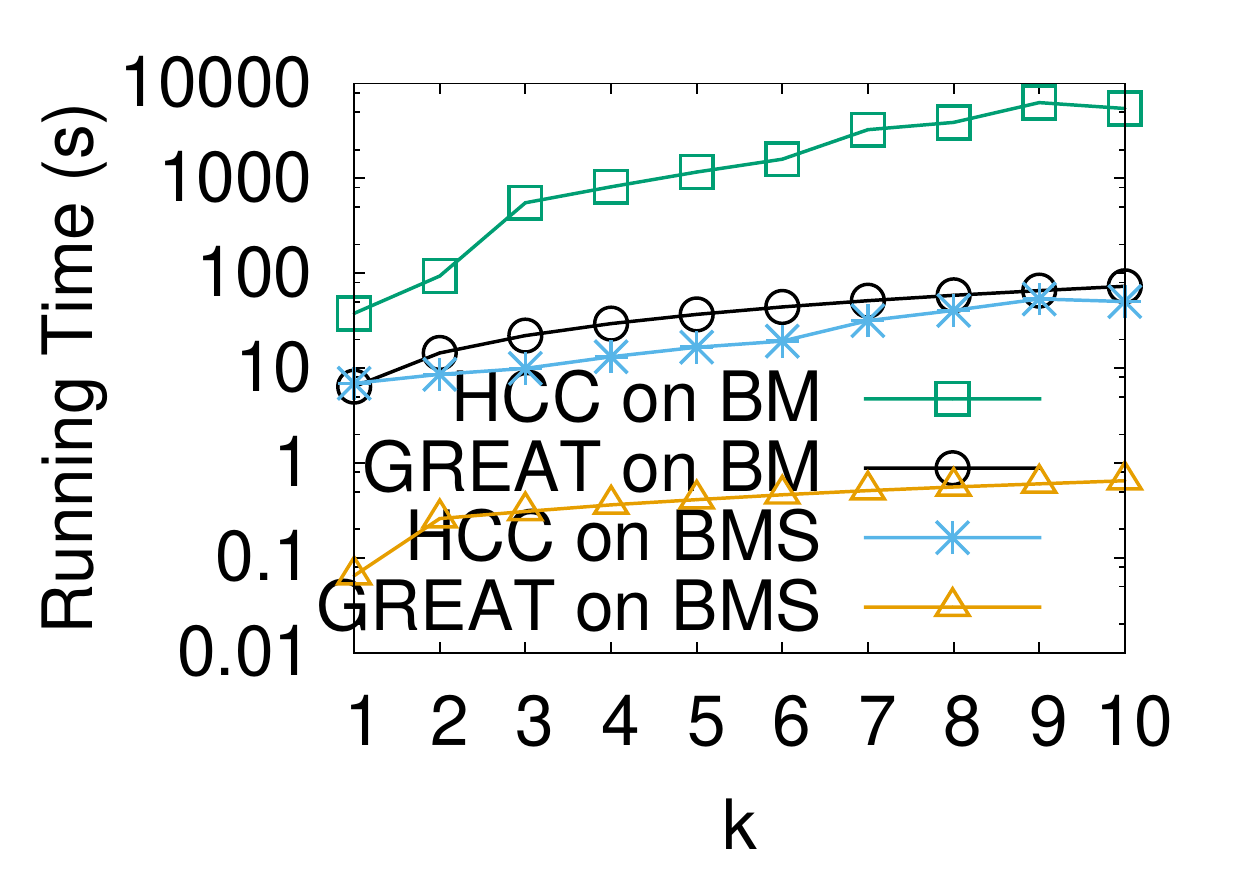}
		\label{subfig:meantips BM ktime}
	}
	\subfloat[Varying $k$ on BLS]
	{
		\includegraphics[width=\subfigwidth]{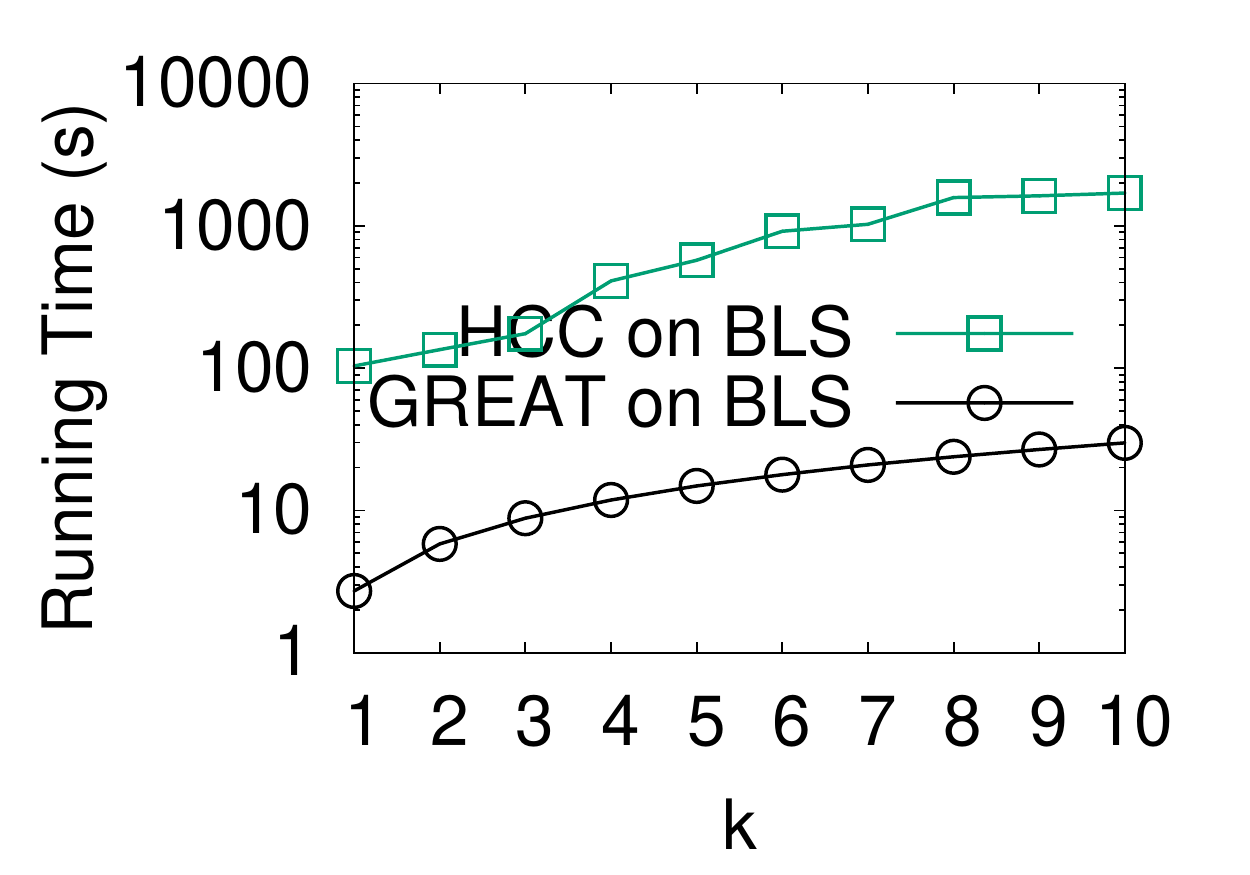}
		\label{subfig:meantips BL k_time}
	}
	\figcaption{\meantips: Running time performance.}
	\label{fig:mean_tips time}
\end{figure}

\subsubsection{Performance Results}
Fig.~\ref{subfig:meantips BM ktime} shows that sampling offers a speed up of
about 2 orders of magnitude on the running times for both \hcc and \inc on BM.
Fig.~\ref{subfig:meantips BL k_time} shows that \inc is about 2 orders of
magnitude faster than \hcc when evaluated on BLS.  \hcc is slower due to its
repeated swaps.  In addition, it is run thrice to avoid a poor random initial
choice.

\subsubsection{Memory Footprint}
The memory consumption of \hcc and \inc on BM at $k=5$ are roughly 18GB and
11.5GB respectively; and, that on BMS are about 3.8 GB and 0.5GB respectively.
The sampling techniques are, thus, highly effective in lowering the memory
footprints of both the algorithms.

\begin{figure}[tb]
	\centering
	\moveups
	\subfloat[Number of sites.]
	{
		\includegraphics[width=\subfigwidth]{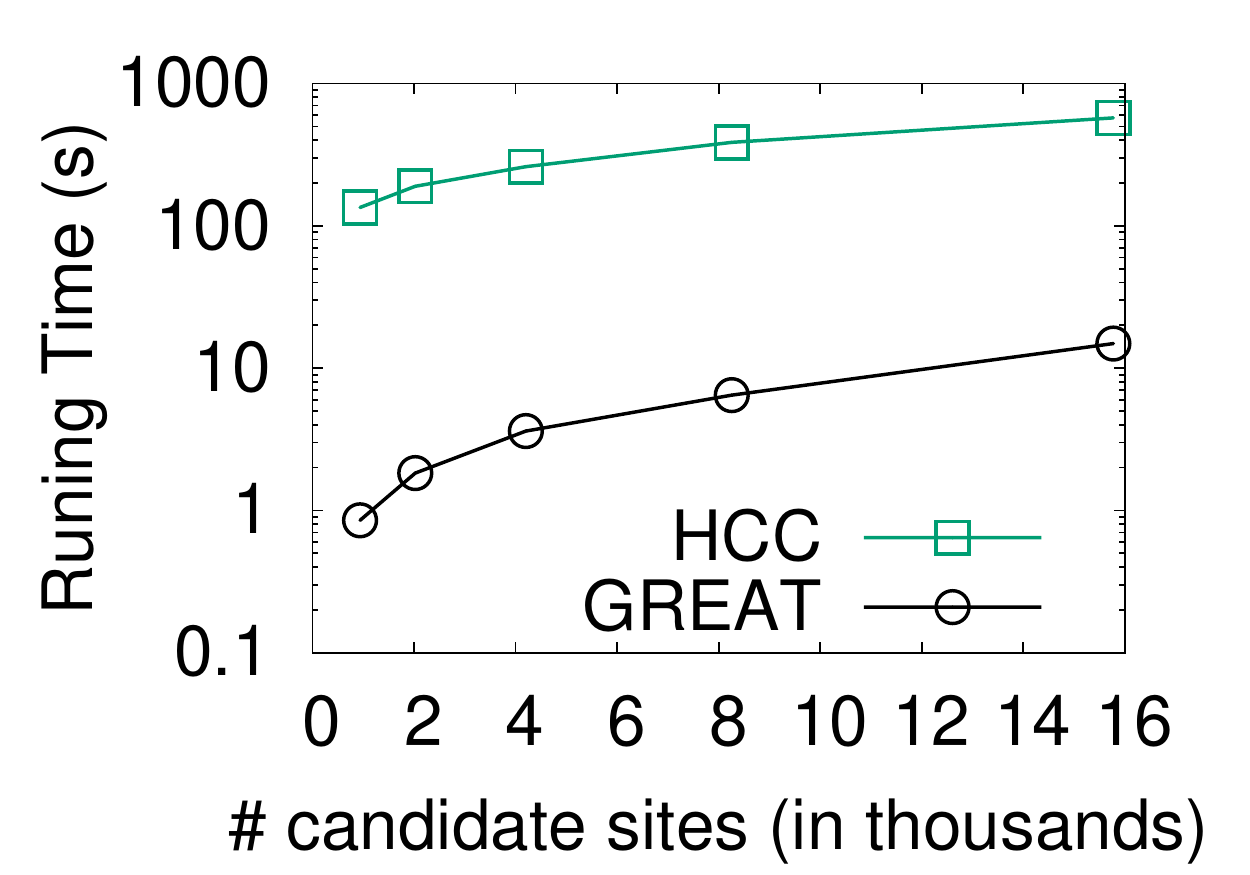}
		\label{subfig:meantips site}
	}
	\subfloat[Number of trajectories.]
	{
		\includegraphics[width=\subfigwidth]{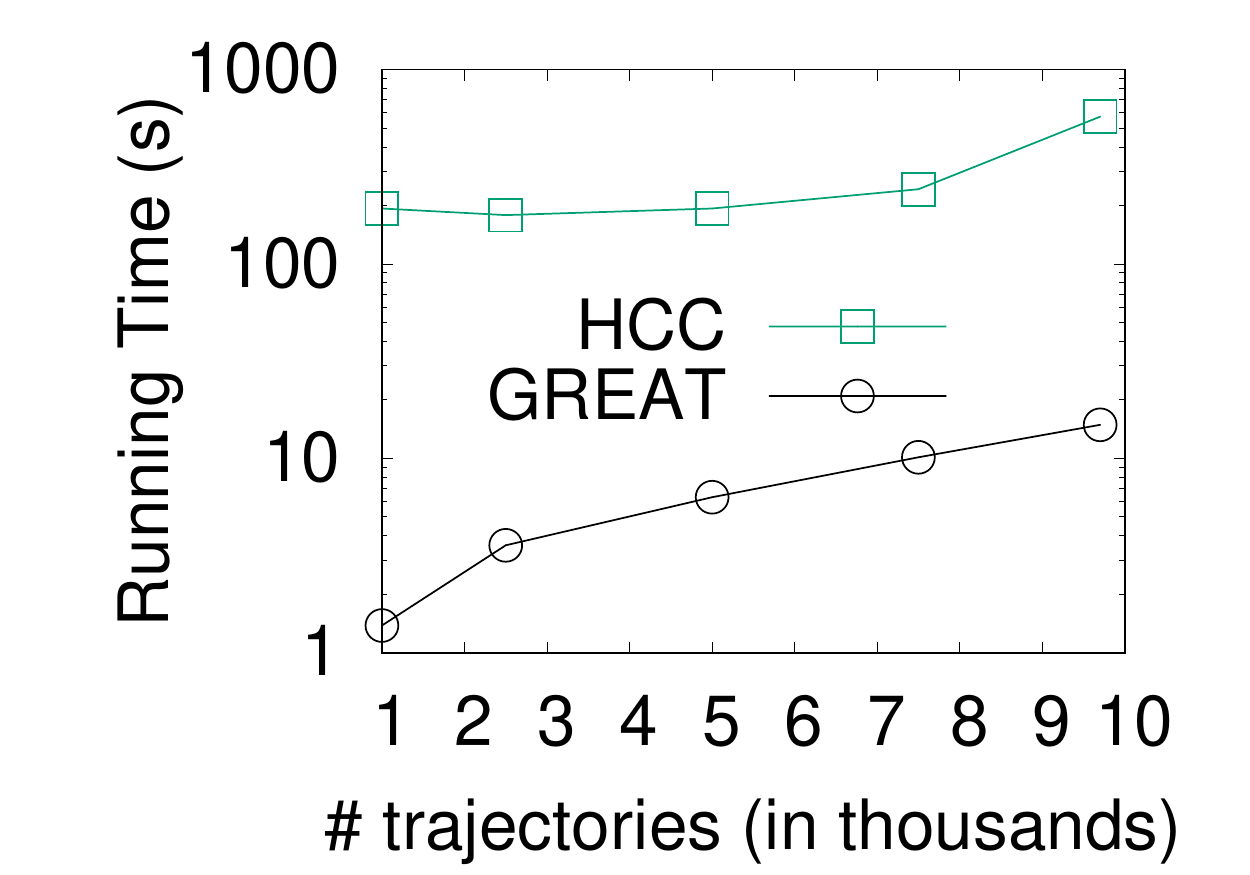}
		\label{subfig:meantips traj}
	}
	\figcaption{\meantips: Scalability results (at $k=5$).}
	\label{fig:nean_tips scale}
\end{figure}

\subsubsection{Scalability}
We next examine the scalability of \hcc and \inc with respect to the number of
candidate sites and trajectories.  We use the same setting as in
Sec.~\ref{sec:maxtips scalability} for \maxtips, although the BM dataset is used
instead of BL.  Fig.~\ref{fig:nean_tips scale} shows that both are scalable with
\inc being faster by about 2 orders of magnitude.

\begin{figure}[bt]
	\centering
	\moveups
	\subfloat[Avg. Inconvenience.]
	{
		\includegraphics[width=\subfigwidth]{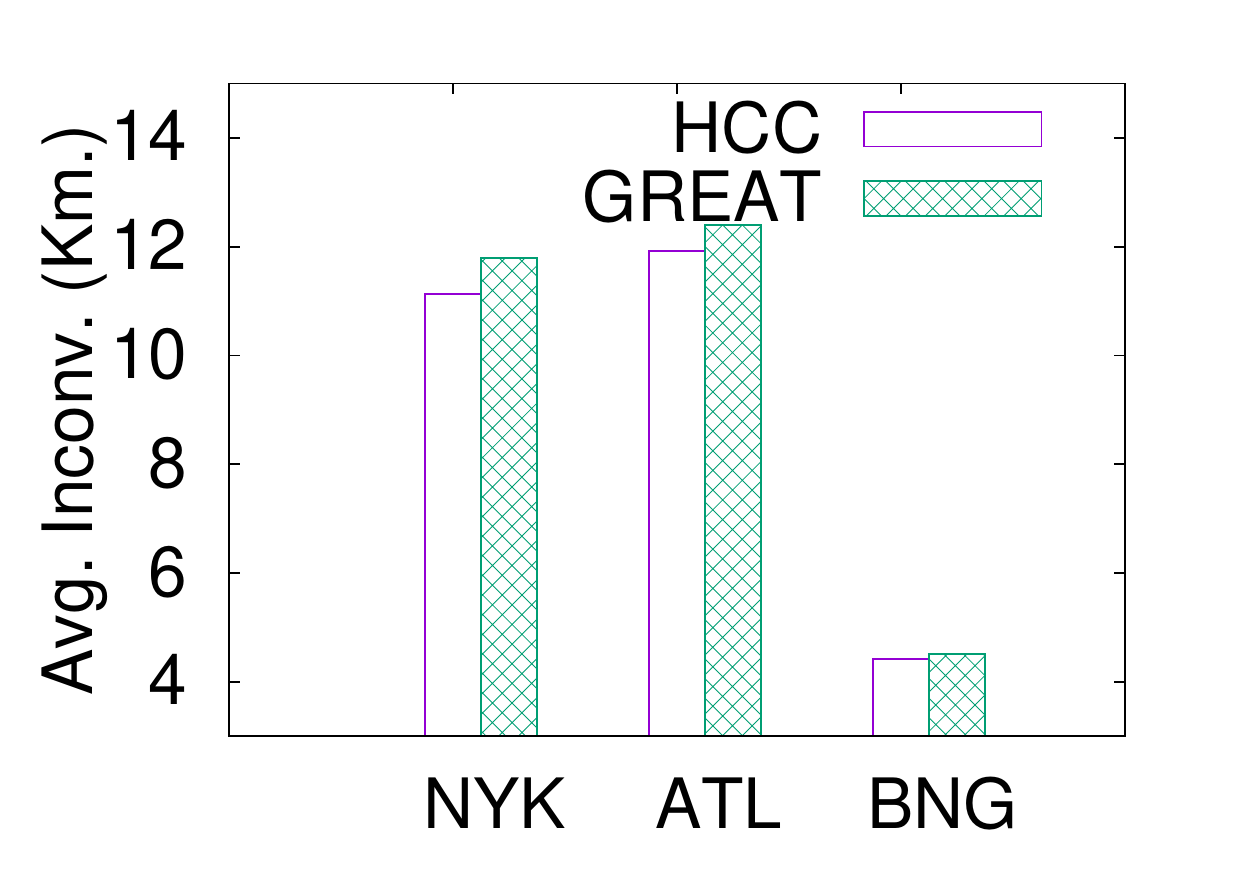}
		\label{subfig:meantips syndist}
	}
	\subfloat[Running time.]
	{
		\includegraphics[width=\subfigwidth]{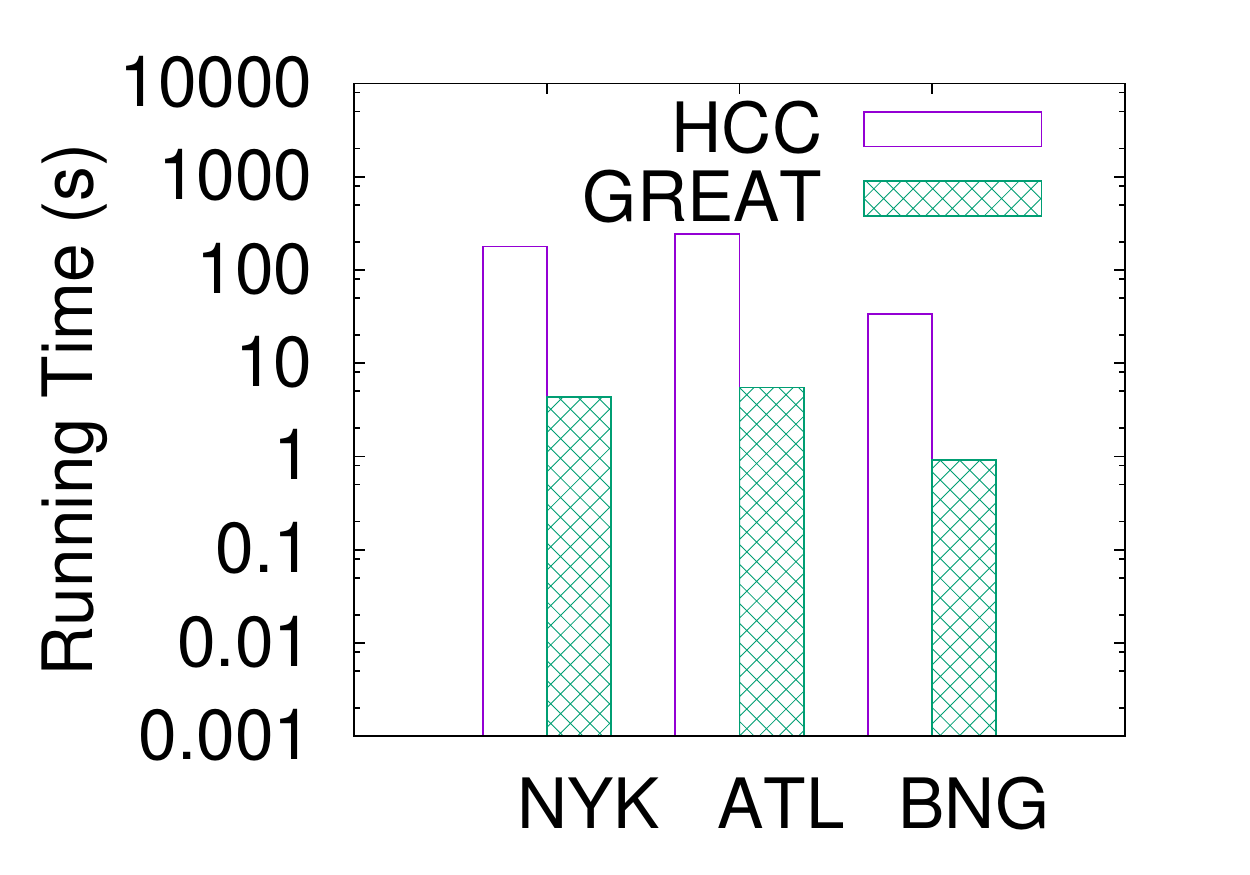}
		\label{subfig:meantips syntime}
	}
	\figcaption{\meantips: Synthetic datasets (at $k=5$).}
	\label{fig:mean_tips synthetic}
\end{figure}

\subsubsection{Synthetic Datasets}
The next result shows the effect on synthetic datasets.  Using the same setup as
discussed in Sec.~\ref{sec:maxtips synthetic}, Fig.~\ref{fig:mean_tips
synthetic} shows that \inc is about 1-2 orders of magnitude faster than \hcc,
while sacrificing no more than 10\% in accuracy over all the three synthetic
datasets. The \ai values for New York and Atlanta are significantly higher than
those of Bengaluru due to their large road network, and even distribution of
trajectories. The running times follow the same trend as in the case of
\maxtips. 

\subsection{Comparison with Baseline} \label{sec:baseline comparison}

\begin{figure}[t]
	\centering
	\moveups
	\subfloat[\maxtips.]
	{
		\includegraphics[width=\subfigwidth]{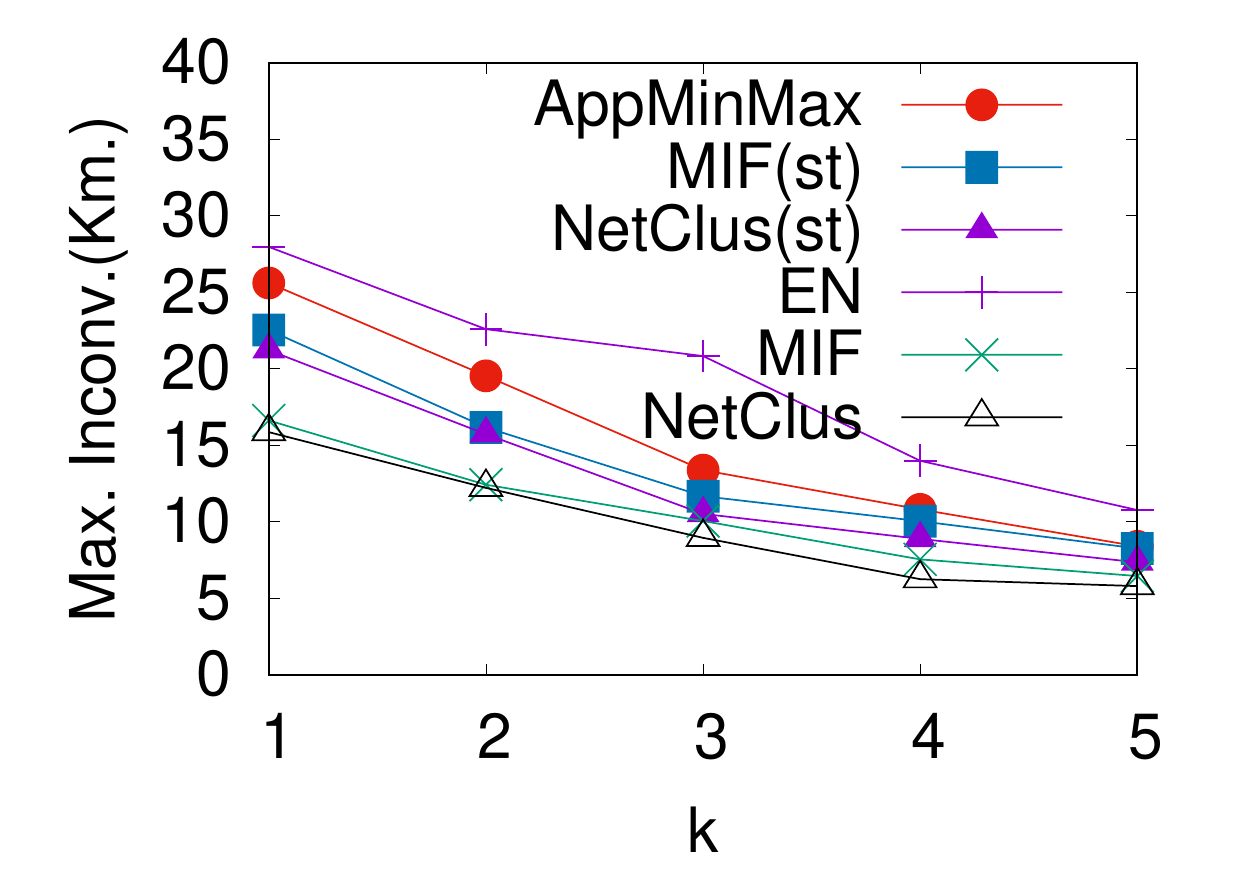}
		\label{subfig:maxtips baseline}
	}
	\subfloat[\meantips.]
	{
		\includegraphics[width=\subfigwidth]{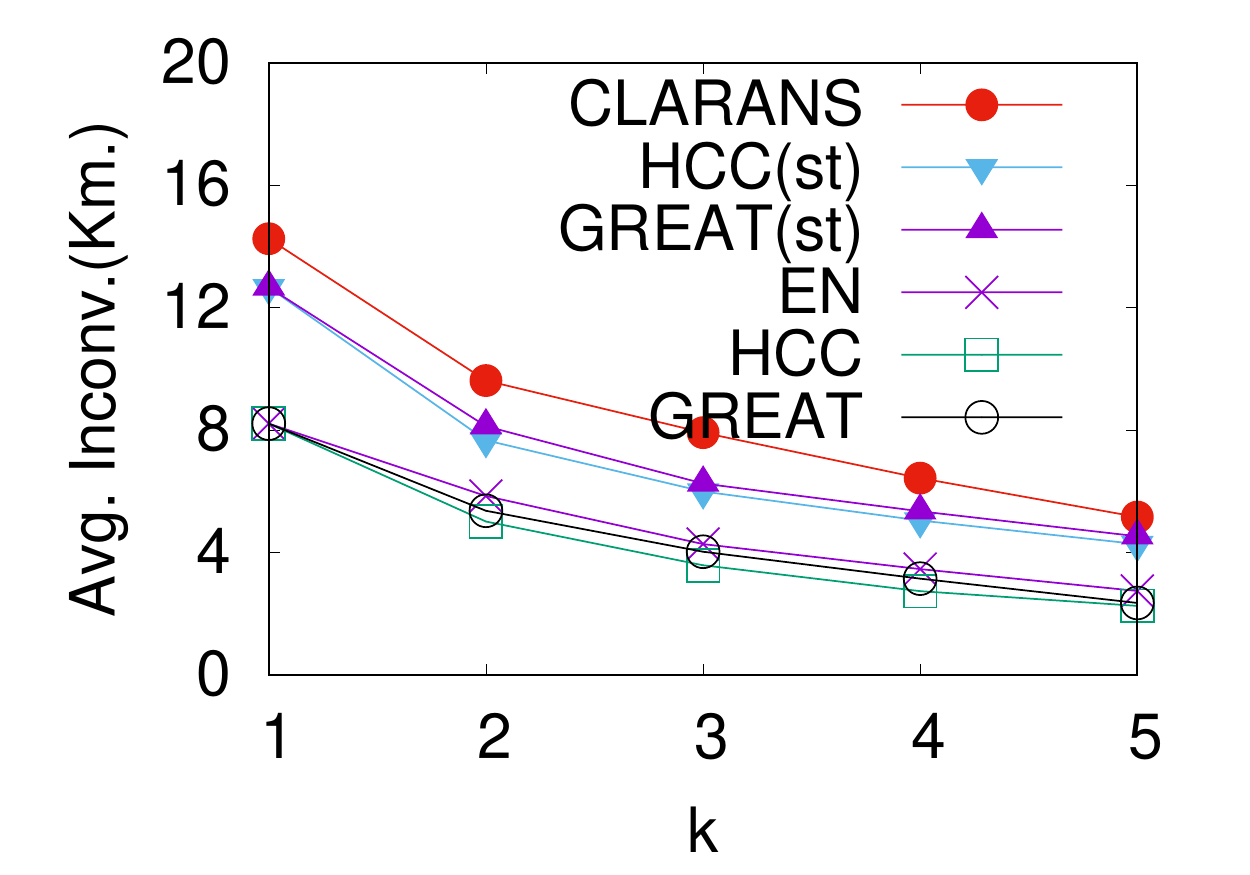}
		\label{subfig:meantips baseline}
	}
	\figcaption{Comparison with Baseline.}
	\label{fig:tips baseline}
\end{figure}

In this section, we compare the performance of our algorithms with the baseline
techniques for static and mobile users. For \maxtips, the baseline techniques
used for static and mobile users are AppMinMax \cite{liu2016finding} and EN
(Extended Network method) \cite{MILERUN} respectively; and for \meantips, the
baseline techniques for static and mobile users are CLARANS
\cite{ng2002clarans} and EN \cite{MILERUN} respectively. Further, we also show
the impact of considering trajectories as opposed to one or more static user
locations such as homes and offices.  Fig.~\ref{fig:tips baseline} shows the
quality comparisons for \maxtips and \meantips.  For \maxtips, we use the BL
dataset, while for \meantips, the BM dataset was used.

To model static users, we use either of the two end-points of a trajectory as
the representative static location of its user. We run AppMinMax over these
static user locations. This algorithm reports $k$ facility locations that
minimize the maximum distance of any user to its nearest facility. To study the
effect of factoring in two static locations of each user, we consider only the
two end-points of each trajectory and ignore the intermediate points. The
corresponding \mif and \nc versions are referred to as MIF(st) and NetClus(st)
respectively.

For mobile users, there is no existing work that aims to minimize the maximum
user-inconvenience. The nearest baseline technique is the EN method that
reports a single facility location that minimizes the average distance of any
user to its nearest facility. To generate $k$ facility locations, we repeat the
algorithm $k$ times, assuming $k-1$ existing facilities. This algorithm along
with \mif and \nc are evaluated over \emph{full} trajectories, i.e., without
skipping any intermediate point in any of the trajectories.  

Fig.~\ref{subfig:maxtips baseline} shows that considering full trajectories is
always much better than considering a single or two static locations per user.
While \nc performs the best, EN performs the worst. This is because EN does not
address the objective of \maxtips, but minimizes the average inconvenience. 

Next, we evaluate \meantips over the BM dataset. Here, the CLARANS algorithm
(discussed in Sec.~\ref{sec:hcc}) serves as the baseline technique for static
users' scenario. To assess the effect of considering only two static locations
per user, we run versions of \hcc and \incg as HCC(st) and GREAT(st)
respectively. The EN method works as the baseline techniques for mobile users. 

As in the case of \maxtips, Fig.~\ref{subfig:meantips baseline} once again
validates our claim that using trajectories is always more beneficial than
considering one or two static locations per user. While CLARANS running over
single static location per user performs the worst, \hcc, running on
\emph{full} trajectories performs the best. The performance of EN is also
good, but marginally lesser than that of \hcc. 

\subsection{Existing Facilities}

\begin{figure}[t]
	\centering
	\moveups
	\subfloat[\maxtips.]
	{
		\includegraphics[width=\subfigwidth]{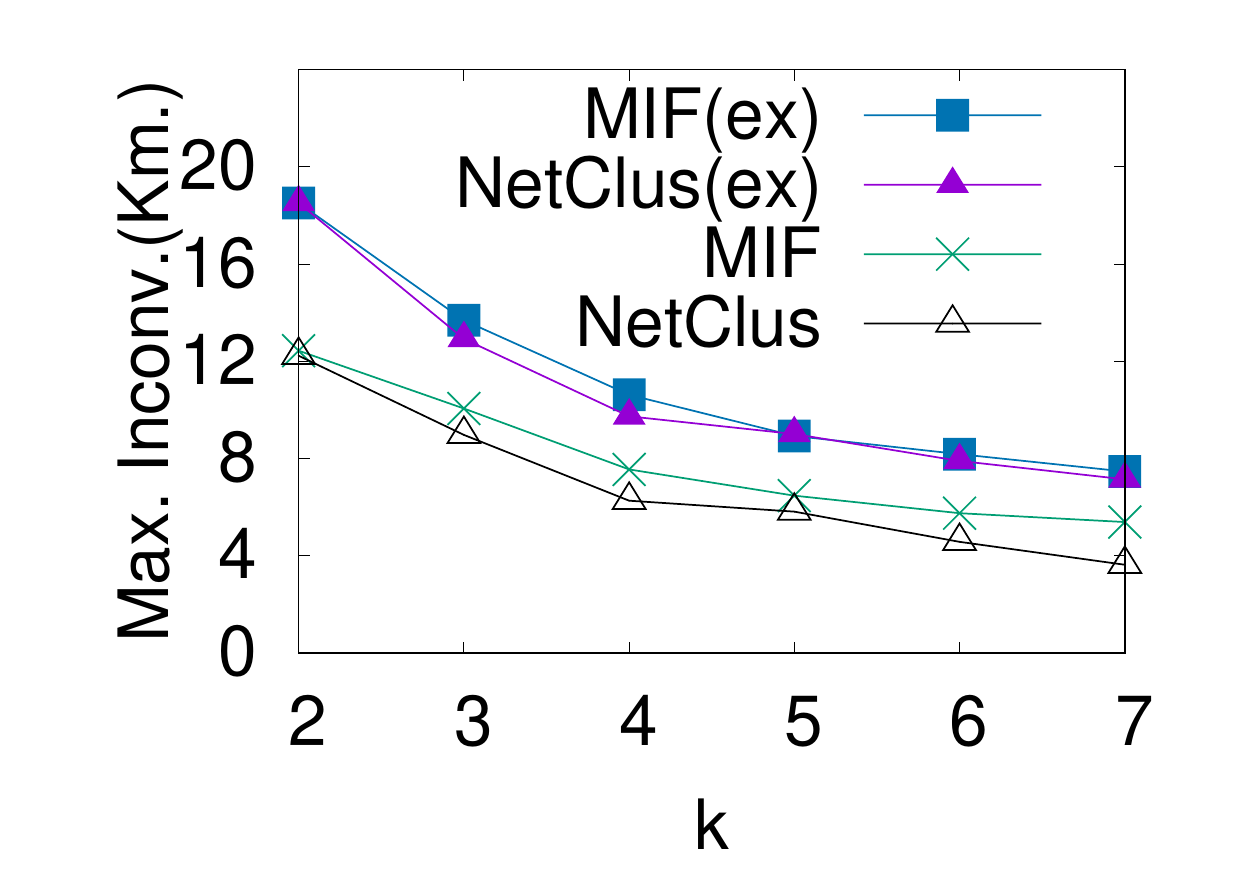}
		\label{subfig:maxtips existing}
	}
	\subfloat[\meantips.]
	{
		\includegraphics[width=\subfigwidth]{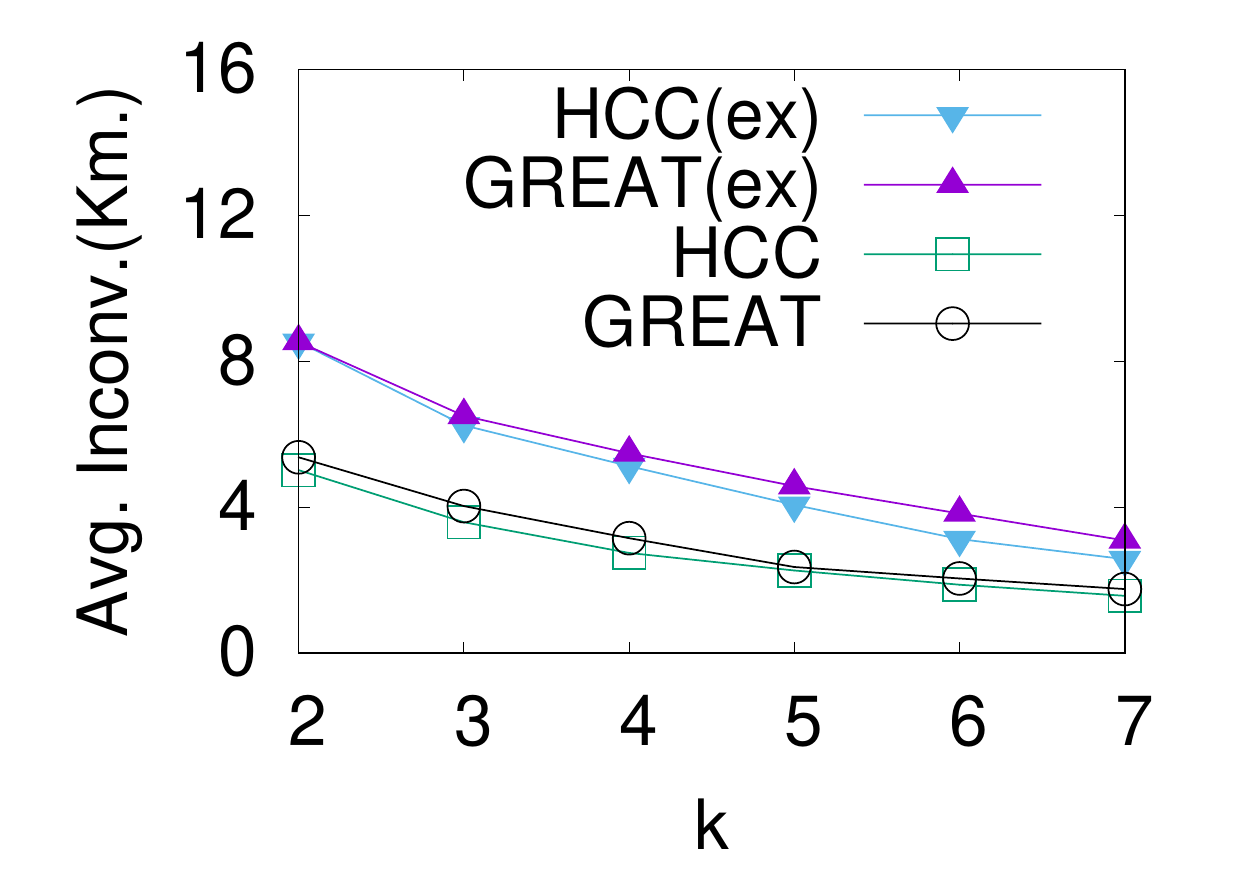}
		\label{subfig:meantips existing}
	}
	\figcaption{Effect of Existing Facilities.}
	\label{fig:tips existing}
\end{figure}

Fig.~\ref{fig:tips existing} shows the effect of existing facilities on \maxtips
and \meantips. We choose $2$ existing facilities randomly. The corresponding
maximum and average inconveniences are $18.51$ Km and $8.53$ Km respectively.
Using the proposed algorithms, we locate $k'$ new facilities for
$k'=1,\cdots,5$. The $k$ values ($2 + k'$) shown in the figure denotes the total
number of facilities including the existing ones. For \maxtips, MIF(ex) and
NetClus(ex) represent the \maxi values as recorded by \mif and \nc respectively,
while factoring in the \emph{existing facilities}. Similarly, for \meantips,
HCC(ex) and GREAT(ex) represent the \ai values as recorded by \hcc and \incg
respectively, while considering the \emph{existing facilities}. The \maxi
(respectively, \ai) values shown by \mif and \nc (respectively, \hcc and \incg)
correspond to the scenario when these algorithms are used to locate new
facilities in the \emph{absence} of any existing facilities.  Since the existing
facilities are chosen randomly, the initial \maxi and \ai values (as mentioned
above) are quite high. As a result, the \maxi values of NetClus(ex) and MIF(ex)
continue to be much higher than those of \nc and \mif respectively, for all
values of $k$. Following the same argument, the \ai values of HCC(ex) and
GREAT(ex) are much higher than those of \hcc and \incg respectively.

\subsection{Summary of Experiments}

We summarize our experimental findings as follows.  \nc offers the best
performance for \maxtips on multiple real and synthetic datasets, both in terms
of efficiency and quality.  It outperforms \mif by more than 20\% in quality and
is faster by an order of magnitude. For \meantips, firstly, we observe that it
is better to apply the sampling techniques when the dataset is very large, such
as \bl. The error in quality due to sampling is reasonably low for both \hcc and
\inc. While \inc is about 2 orders of magnitude faster than \hcc, its quality is
about 10\% lower than \hcc.  Thus, if the goal is to achieve low average
inconvenience regardless of high but practical running times, one may choose
\hcc. On the other hand, if a fast query time is desired with reasonably high
accuracy in quality, \inc may be chosen.

%% file: conc.tex
\section{Conclusions}
\label{sec:conc}

In this paper, we introduced two facility location problems over
user-trajectories, namely, \maxtips and \meantips, that aim to minimize the
maximum and the average user-inconvenience, respectively. We showed that both
these problems are NP-hard and proposed one optimal algorithm and two efficient
heuristics for each of them.  The heuristics can work both in the presence or
absence of existing facilities.  Empirical evaluation over large-scale real and
synthetic datasets show that the proposed solutions are effective in terms of
quality and efficient in terms of space and running time.

In future, we will explore other trajectory-based facility location problems.

%% file: appendix.tex
\section{Proofs of Theorems}

\subsection{Proof of Theorem~\ref{thm:mif approx bound} (Quality of \mif)}
\label{sec:mif approx bound}
 
\begin{proof}
	Let $\q=\{s_1,\dots,s_k\}$ be the set of $k$ sites returned by \mif, with a
	maximum inconvenience of $d$.  Let $\reptraj=\{T_1,\dots,T_k\}$ be the $k$
	representative trajectories chosen by \mif.  Assuming $\s=V$, it follows
	that each node in $T_i$ is a candidate site.  Since $s_i$ is the nearest
	candidate site to $T_i$, it must be that $s_i \in T_i$.

	Now, let us consider the case $k=1$.  Let $d$ be the maximum inconvenience
	due to the selection of the site $s_1 \in T_1$.  Now suppose $s_1^*$ is the
	site reported by an optimal algorithm with optimal maximum inconvenience
	value $d^*$. Hence, $\forall T_q \in \traj, \ \dr(T_q,s_1^*) \le d^*$.
	Therefore, there exists a site $s_1^\prime \in T_1$, such that
	$\dr(T_q,s_1^\prime) \le 2d^*$ for any trajectory $T_q \in \traj$, because
	there exists a path  from $s_1^\prime$ via $s_1^*$ to  trajectory $T_q $ of
	distance at most $2d^*$. Thus, the maximum inconvenience $d$ due to the site
	$s_1$ is  $d \le 2d^*$. 

	Now, consider the case $k \ge 2$.  Consider an optimal solution
	$\q^*=\{s_1^*,\dots,s_k^*\}$ with maximum inconvenience $d^*$. Let
	$\clus(s_i^*)$ denote the set of trajectories served by $s_i^*$.  Using the
	pigeon-hole principle, we conclude that at least two representative
	trajectories $T_p, T_q \in \reptraj, p<q$, must belong to one of the $k$
	clusters, say $\clus(s_i^*)$. Since the maximum inconvenience of the
	solution $\q$ is $d$, therefore, $\dr(T_q,s_p) \ge d$. Moreover, since
	$T_p,T_q \in \clus(s_i^*)$, we get $\dr(T_p,s_i^*)\le d^*, \ \dr(T_q,s_i^*)
	\le d^*$.  From this, we infer that there must be a site $s_p^\prime \in
	T_p$ such that $\dr(s_p^\prime,s_i^*) \le d^*$ and, therefore,
	$\dr(T_q,s_p^\prime) \le 2d^*$. Since the length of any trajectory is at
	most $L$, $\dr(s_p,s_p^\prime) \le L$. From this, we conclude that $d \le
	\dr(T_q,s_p) \le \dr(T_q,s_p^\prime)+\dr(s_p^\prime,s_p) \le 2d^*+L$. 
	\hfill{}
\end{proof}

\subsection{Proof of Theorem~\ref{thm:mif complexity} (Complexity of \mif)}
\label{sec:mif complexity proof}

\begin{proof}
	We assume that the number of road segments (edges) is $O(n)$ as the road
	networks are roughly planar. Thus, from any node $v_i\in V$, the distances
	to all other nodes in the network can be computed in $O(n\log n)$ time,
	using Dijkstra's shortest path algorithm \cite{algorithmsbook}.  
	
	We analyze the computation cost of any iteration as follows.  Recall that
	\nn map stores the trajectories in a sorted order based on their distance to
	the nearest facility in \q. Since $|\nn|=m$, hence to identify the
	trajectory $T_i$ at rank \trajthres, requires $O(1)$ time, using array
	implementation of \nn. 
	
	Then, for each node $v \in T_i$ (that $T_i$ passes through), the distances
	are computed to all other nodes in $V$. If the maximum number of nodes in
	any trajectory is $l$, i.e., $|T_i| \le l$, then the above distance
	computation step takes $O(l.n\log n)$ time. Thus, identifying the nearest
	candidate site to $T_i$, say $s_i$, requires $O(l.n\log n)$ time.  

	Following this, $s_i$ is added to \q. Then the distances are computed
	between $s_i$ and all other trajectories in \traj. To do this, we first
	compute distances between $s_i$ and all nodes in $V$, which requires
	$O(n\log n)$ time. Assume these distances are indexed on site-ids.  The
	distance between $s_i$ and any trajectory can be computed in $O(l^2)$ time,
	as there are at most $O(l^2)$ distance look-ups. Therefore, the distance
	between $s_i$ and all trajectories in \traj can be computed in $O(ml^2)$
	time. If the distance  of any trajectory $T_j \in \traj$ to its nearest
	facility in \q is more than that with $s_i$, then this value is updated in
	the \nn map.  This updation step takes $O(m)$ time over all the
	trajectories. Sorting the \nn map takes $O(m \log m)$ time. Summing up all
	these costs, over each of the total $k$ iterations, the total time
	complexity is $O(k.l.n\log n + k.m.l^2 + k.m\log m)$.

	Next, we analyze the space complexity.  The road network and candidate sites
	can be stored in $O(n)$ space. Storing the trajectories in \traj require
	$O(l.m)$ space.	Maintaining the \nn map requires $O(m)$ space.  During each
	iteration of the algorithm, the node-to-node distances are computed for each
	node of the previously chosen representative trajectory.  Storing these
	distance values require at most $O(ln)$ space. As these distance values are
	no longer used in subsequent iterations, they are discarded at the end of
	each iteration. Thus, the maintenance overhead of the node-to-node distances
	over the $k$ iterations is $O(l.n)$. Storing the sets \reptraj and \q
	require $O(k)=O(n)$ space. Hence, the total space complexity is $O(l(n+m))$.
	\hfill{}
\end{proof}

\subsection{Proof of Theorem~\ref{thm:netclus complexity} (Complexity of \nc)}
\label{sec:netclus complexity proof}

\begin{proof}
	Assuming the time required to answer a \tops query by \nc to be
	$O(t_{\tops})$, since $O(\log_2(\tau_{max}/\tau_{min}))$ \tops queries are
	executed, the total time is $O(\log_2(\tau_{max}/\tau_{min}). t_{\tops})$.
	The time complexity, $O(t_{\tops})$, of \nc is analyzed in
	\cite{tops_paper}, and is beyond the scope of this paper.

	Now, let us analyze the space complexity of \nc. As discussed above, if the
	index resolution parameter is $\eps>0$, then the total number of index
	instances is $t=1+\lfloor\log_{1+\eps}(\tau_{max}/\tau_{min})\rfloor$. For
	each index instance, the nodes in $V$ are divided into clusters which
	requires $O(n)$ space. In addition, for each cluster, we store the set of
	trajectories that pass through it. This cost is $O(m.l)$ where $m=|\traj|$
	and $l$ is the largest number of nodes in any trajectory. For each node in
	$V$, we store its distance to the cluster-center of the cluster it belongs
	to. This storage cost is $O(n)$ across all the nodes. Further, for each
	trajectory, we store its distance to the cluster-center of the clusters it
	passes through. Since a trajectory cannot pass through more than $l$
	clusters, this storage cost over all the trajectories is no more than
	$O(m.l)$. Therefore, the total space complexity is $O(t.(n+m.l))$.
	\hfill{}
\end{proof}

\subsection{Proof of Theorem~\ref{thm:hcc complexity} (Complexity of \hcc)}
\label{sec:hcc complexity}

\begin{proof}
	Let $n^\prime,m^\prime$ denote the number of sites and trajectories, that
	are produced after sampling of the sites and the trajectories, respectively.
	Thus, $n^\prime \le n$ and $m^\prime \le m$. We assume that $k \le
	n^\prime$.
	
	The algorithm is executed for $t$ trials, and in each trial, the maximum
	possible number of iterations is $\eta$. In a given iteration, the number of
	swaps that are scanned is $\nfs . k(n^\prime-k)$. To evaluate each swap, we
	need to compute the total inconvenience of all the $m^\prime$ trajectories
	with respect to the current set of $k$ medoids. This requires
	$O(k.m^\prime)$ time (assuming that the distance between each pair of site
	and trajectory is pre-computed). Hence, the total running time is
	$O(t.\eta.\nfs.k^2(n^\prime-k).m^\prime)$ time.

	Next, let us analyze the space complexity. Storing the sampled sites and
	trajectories require $O(n^\prime)$ and $O(m^\prime l)$ space where $l$ is
	the maximum number of nodes in any trajectory. Storing the pairwise distance
	between each pair of site and trajectory in the sampled space requires
	$O(m^\prime.n^\prime)$ space. To store the \nn maps, we need $O(m^\prime)$
	space. In order to store the current set of medoids, we need
	$O(k)=O(n^\prime)$ space. Summing up all the costs, we find that the total
	space complexity is $O(m^\prime.(n^\prime +l))$.
	\hfill{}
\end{proof}

\subsection{Proof of Theorem~\ref{thm:bound of incg for meantips} (Quality of \incg)}
\label{sec:great approx bound}

\begin{proof}
	It is easy to see that \incg returns the optimal solution for $k=1$ as it
	performs an exhaustive search over all the candidate sites in \s.
	
	Now, suppose $k \ge 2$.  For any given set of candidate sites $\q \subseteq
	\s$, consider a function $f(\q)=\totali(\fac)-\totali(\q)$.  Since
	$\totali(\q)$ is a non-increasing super-modular function
	(Th.~\ref{thm:supermodular}), it follows that $f(\q)$ is a non-decreasing
	sub-modular function. Further, if the existing set of facilities $\fac \neq
	\varnothing$, then the initial total inconvenience for $\q =\varnothing$,
	can be written as $\totali(\varnothing)=\totali(\fac)$. Thus,
	$f(\varnothing)=\totali(\fac)-\totali(\varnothing)=0$. It is known that the
	greedy heuristic offers an approximation bound of $1-1/e$ for any
	non-decreasing sub-modular function $f$ with $f(\varnothing)=0$
	\cite{nemhauser1978analysis}. Since $OPT$ is a set that minimizes $\totali$,
	from the definition of $f$, it follows that $OPT$ must maximize $f$ as
	$\totali(\fac)$ is a constant. Let $\q$ be the set of sites reported by
	\incg. Since \incg essentially mimics the greedy heuristic for
	non-decreasing sub-modular functions \cite{nemhauser1978analysis}, the same
	approximation bound is applicable for $f$. Therefore, $f(\q) \ge
	(1-1/e)f(OPT)$. From the definition of $f$, we can write $\totali(\q) \le
	(1-1/e)(\totali(OPT)-\totali(\fac))+\totali(\fac) \le
	(1-1/e)\totali(OPT)+\totali(\fac)/e$.	
	\hfill{}
\end{proof}

\subsection{Proof of Theorem~\ref{thm:complexity of incg for meantips} (Complexity of \incg)}
\label{sec:complexity of incg for meantips}

\begin{proof}
	The space required to store the trajectories and the nodes on the road
	network are $O(ml)$ and $O(n)$ respectively. Further, storing the distance
	values for each pair of site $s_i \in \s$ and trajectory $T_j \in \traj$
	requires $O(m.n)$ space.  Further, storing the $\nn()$ maps require $O(m)$
	space. Since $|\q| =k =O(n)$, the overall space complexity is $O(m.(n+l))$.
	
	We next analyze the time complexity. In each iteration of the \incg
	algorithm, we compute the total inconvenience $\totali(\q \cup \{s_i\})$ for
	each site $s_i \in \s \setminus \q$.  For each site $s_i$, this computation
	requires $O(m)$ time over all the trajectories. Once a site is added to
	$\q$, the $\nn$ maps are updated in $O(m)$ time. Thus, each iteration
	requires $O(m.n)$ time. The overall time complexity of \incg, running over
	$k$ iterations is, therefore, $O(k.m.n)$.
	\hfill{}
\end{proof}